\documentclass[%
 reprint,
superscriptaddress,
nofootinbib,
 amsmath,amssymb,
 aps,
 pra,
 floatfix,
 showkeys
]{revtex4-2}
\pdfoutput=1

\usepackage{microtype}
\usepackage{graphicx}
\usepackage{dcolumn}
\usepackage{bm}
\usepackage[%
pdfauthor={Samo Novak, Raul Garcia-Patron},
pdftitle={Laplace expansions and tree decompositions: A faster polytime algorithm for shallow nearest-neighbour Boson Sampling},
breaklinks
]{hyperref}
\usepackage{breakurl}
\usepackage{xcolor}
\hypersetup{
    colorlinks,
    linkcolor={red!50!black},
    citecolor={blue!50!black},
    urlcolor={blue!80!black}
}
\usepackage[capitalize]{cleveref}

\usepackage{xcolor}
\usepackage{inputenc}
\usepackage{wasysym}
\usepackage{dsfont}
\usepackage{amsthm}
\usepackage{mathtools}
\usepackage{cancel}
\usepackage{extarrows}
\usepackage{enumitem}
\usepackage[caption=false]{subfig}
\usepackage[ruled,vlined,linesnumbered]{algorithm2e}
\usepackage{float}
\usepackage{tabularray}

\newtheorem{theorem}{Theorem}[section]
\newtheorem{lemma}[theorem]{Lemma}

\theoremstyle{definition}
\newtheorem{definition}[theorem]{Definition}

\newtheorem{example}[theorem]{Example}
\newtheorem{note}[theorem]{Note}

\renewcommand{\vec}[1]{\underline{#1}}
\newcommand{\ket}[1]{\ensuremath{\left|#1\right\rangle}}
\newcommand{\bra}[1]{\ensuremath{\left\langle#1\right|}}
\newcommand{\neket}[1]{\ensuremath{|#1\rangle}}
\newcommand{\nebra}[1]{\ensuremath{\langle#1|}}

\newcommand{\braoket}[3]{\ensuremath{\left\langle#1\middle|#2\middle|#3\right\rangle}}
\newcommand{\restr}[2]{\ensuremath{\left.\kern-\nulldelimiterspace #1
      \vphantom{\big|}
    \right|_{#2}}}

\newcommand{\C}{\mathbb{C}}

\newcommand{\N}{\mathbb{N}}
\newcommand{\cO}{\mathcal{O}}
\newcommand{\tcO}{\widetilde{\!\mathcal{O}}}
\newcommand{\Sg}{\mathrm{S}}
\newcommand{\I}{\mathds{1}}
\newcommand{\U}{\mathcal{U}}
\newcommand{\V}{\mathcal{V}}
\newcommand{\W}{\mathcal{W}}
\renewcommand{\Pr}{\mathbb{P}}
\newcommand{\Expectation}{\mathbb{E}}

\newcommand{\Tds}{\bm{\mathcal{T}}}
\newcommand{\Rows}{\mathcal{R}}
\newcommand{\Cols}{\mathcal{C}}

\newcommand{\deq}{\coloneqq}
\DeclareMathOperator{\per}{per}

\DeclareMathOperator{\tw}{tw}
\newcommand{\twG}{\widehat{\tw}}
\DeclareMathOperator{\ch}{ch}
\DeclareMathOperator{\rt}{root}

\newcommand{\captionitem}[1]{\hspace*{\parindent}\textbf{(#1)}}

\newcommand{\orcid}[1]{\raisebox{-0.1ex}{\includegraphics[width=6pt]{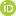}}
  \href{https://orcid.org/#1}{#1}
}

\AddToHook{cmd/appendix/before}{%
  \crefalias{section}{appendix}%
  \crefalias{subsection}{appendix}%
  \crefalias{subsubsection}{appendix}
}

\begin{document}


\title{Laplace expansions and tree decompositions: A faster polytime algorithm for shallow nearest-neighbour Boson Sampling}

\author{Samo Nov\'{a}k}
\thanks{\orcid{0000-0002-2713-9593}}
\email{samo.novak@inria.fr}
\affiliation{%
  Laboratory for the Foundations of Computer Science,
  School of Informatics,
  University of Edinburgh
}%
\affiliation{%
  Mathematical Institute,
  University of Oxford
}%
\affiliation{%
  ORCA Computing, London, UK
}%
\affiliation{%
  Inria Paris, France
}%
\author{Ra\'{u}l Garc\'{i}a-Patr\'{o}n}
\thanks{\orcid{0000-0003-1760-433X}}
\email{rgarcia3@exseed.ed.ac.uk}
\affiliation{%
  Laboratory for the Foundations of Computer Science,
  School of Informatics,
  University of Edinburgh
}%
\affiliation{%
  Phasecraft Ltd., London, UK
}%

\date{24 December 2024, Updated: 14 January 2026}

\begin{abstract}
  In a Boson Sampling quantum optical experiment we send $n$ individual photons into an $m$-mode interferometer and we measure the occupation pattern on the output.
  The statistics of this process depending on the permanent of a matrix representing the experiment, a \#P-hard problem to compute, is the reason behind ideal and fully general Boson Sampling being hard to simulate on a classical computer.
  We exploit the fact that for a nearest-neighbour shallow circuit, i.e. depth $D = \mathcal{O}(\log m)$, one can adapt the algorithm by Clifford \& Clifford~\cite{Cliffords2018} to exploit the sparsity of the shallow interferometer using an algorithm by Cifuentes \& Parrilo~\cite{Cifuentes2015} that can efficiently compute a permanent of a structured matrix from a tree decomposition.
  Our algorithm generates a sample from a shallow circuit in time $\mathcal{O}(n^22^\omega \omega^2) + \mathcal{O}(\omega n^3)$, where $\omega$ is the treewidth of the decomposition which satisfies $\omega \le 2D$ for nearest-neighbour shallow circuits.
  The key difference in our work with respect to previous work using similar methods is the reuse of the structure of the tree decomposition, allowing us to adapt the Laplace expansion used by Clifford \& Clifford which removes a significant factor of $m$ from the running time, especially as $m>n^2$ is a requirement of the original Boson Sampling proposal.%
  \footnote{Accepted in PRA on 6 February 2026 as \url{https://doi.org/10.1103/l4zf-ls9v}.}
\end{abstract}


\maketitle

\section{\label{sec:intro}Introduction}

\begin{figure*}[t]
  \centering
  \includegraphics[scale=0.96]{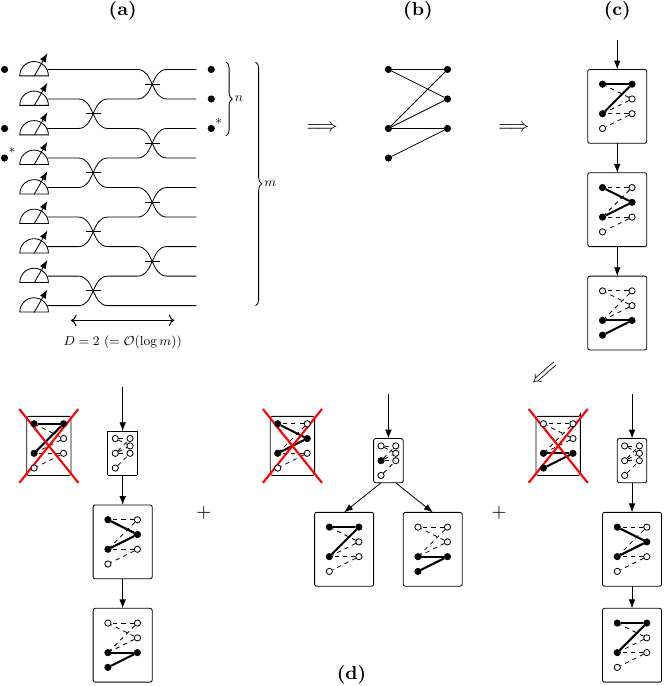}
  \caption{Overview of our work using an example:
    \captionitem{a}~The shallow brickwork circuit of nearest-neighbour beamsplitters with inputs on right and outputs on the left, showing one \textbf{measurement outcome} with previously sampled photons ``$\bullet$'' and a new photon ``$\bullet^*$'' added in this step of the Clifford and Clifford algorithm to expand the sample.
    \captionitem{b}~A \textbf{weighted bipartite graph} represents the outcome: vertices are input and output photons, and non-zero transition amplitudes between them give weighted edges. The outcome amplitude is the sum of weighted perfect matchings in this graph, i.e. the permanent.
    \captionitem{c}~To exploit the \textbf{sparsity} structure, we construct a \textbf{tree decomposition}: a tree whose nodes contain sets of edges of the graph in (b), subject to simple rules. We use a linear decomposition (no branching) where each node corresponds uniquely to an input photon vertex.
    We then modify the tree decomposition in several steps, each time temporarily removing a different input photon $=$ column $=$ node of the tree,  replacing the latter by a dummy with no new information, seen in~\textbf{(d)}. Permanents of these modified trees form the \textbf{Laplace expansion}, our main technical contribution. When done in the right order, conceptualized as a \textbf{``machine head''} walking the tree from one end to the other, we can reuse most of the dynamic programming tables required by the Cifuentes and Parrilo permanent algorithm. As seen in the middle term, we note that the dummy may need to be nonempty.}
  \label{fig:intro}
\end{figure*}

\emph{Boson Sampling} was introduced by Aaronson \& Arkhipov~\cite{Aaronson_Arkhipov2013} in 2010 to propose a suitable candidate for the experimental demonstration of \emph{quantum advantage}, which is when a quantum computer can efficiently solve a problem that a classical computer cannot~\cite{Preskill2012}.
Boson Sampling consists of sending $n$ fully indistinguishable photons through an interferometer acting on $m$ modes, and detecting the photons as they come out at the output of the interferometer. The work by Aaronson and Arkhipov provided strong evidence that such a quantum optics circuit would be hard to simulate on a classical computer when the interferometer was chosen at random from the Haar ensemble.

The intuition behind the hardness of Boson Sampling relies in the fact that its output probabilities require computing permanents of matrices that scale with the number of photons in the interferometer.
The permanent of a matrix $A \in \C^{n \times n}$ is defined as
\[ \per A = \sum_{\sigma \in \Sg_n} \prod_{i=1}^n A_{i, \sigma(i)}, \]
where $\sigma$ is a permutation from the symmetric group on~$n$ elements $\Sg_n$.
Despite its similarity with the definition of the determinant (without switching signs), computing the permanent exactly or approximating it for matrices with complex entries is \#P-hard~\cite{Valiant79}. The best known algorithms have exponential running time of $\cO(n2^n)$: these are due to Ryser~\cite{MR0150048} and Glynn~\cite{Glynn2010}.

A celebrated theoretical result was the sampling algorithm by Clifford and Clifford that can generate an $n$-photon Boson Sampling output with a computational cost that is, up to a factor of 2, equivalent
to the computation of the permanent of a matrix of size $n$. Specifically, the running time of their algorithm is $\cO(n 2^n) + \cO(mn^2)$~\cite{Cliffords2018}.

The prospect of performing an experimental demonstration of quantum advantage attracted a lot of experimental effort and further theoretical work, which recently lead to experiments claiming they have reached quantum advantage on photonic platforms~\cite{Pan2020,Pan2021}.
It is known that one can design classical algorithms that exploit the presence of imperfections in the hardware, such as losses and distinguishability of photons, to run faster and potentially challenge the quantum advantage claims. For example, it was known that for integrated photonic circuits, composed of nearest-neighbour gates with transmission rates decaying exponentially with the depth of the circuit, circuits of depth $D = \Omega(\log m)$
may be simulated as thermal noise in polynomial time, and that a noisy shallow circuit can be simulated in quasi-polynomial time using tensor networks~\cite{Garcia-Patron2019}.

The fact that simulating a shallow circuit, i.e. with depth $D = \cO(\log m)$, requires quasi-polynomial time was expected to be a feature of the use of tensor network techniques, rather than a fundamental result. It was believed that one could find a polynomial time algorithm, for example using results about efficient computation of permanents of sparse matrices by Cifuentes \& Parrilo~\cite{Cifuentes2015}. Indeed, this combination was used by Oh et al. in~\cite{Oh2021} to prove some complexity results about the (non)-hardness of Boson Sampling circuits with a lattice structure. Furthermore, an efficient algorithm for shallow Gaussian Boson Sampling circuits, a variant of Boson Sampling using squeezed states instead of single photons, was proven by Qi et al. in~\cite{qi2020} using similar ideas.

In this work we revisit the inital conjecture in~\cite{Garcia-Patron2019} and show how one can use the result of Cifuentes and Parillo~\cite{Cifuentes2015} to design a polynomial time algorithm for Boson Sampling from shallow circuits, exploiting the fact that the computation of a permanent scales polynomially in the treewitdh of the graph corresponding to our photonic circuit of interest.
A similar approach was given by Oh et al.~\cite{Oh2021} who gave the running time of $\cO(m n^2 2^\omega \omega^2)$.

We present an algorithm that reproduces all the key advantages of the initial Clifford and Clifford algorithm in~\cite{Cliffords2018} adapted to the nearest-neighbors shallow circuit scenario.
Our main technical contribution is the following: We adapt the use of sparsity of matrices in the algorithm by Cifuentes and Parillo~\cite{Cifuentes2015} to combine it with the use of the Laplace expansion in the algorithm by Clifford and Clifford. We do this in a time-efficient manner to remove the dependency on the number of modes in~\cite{Oh2021}, achieving a running time of $\cO(n^22^\omega \omega^2) + \cO(\omega n^3)$.

\subsection{Informal summary}

Here, and in the running example of \Cref{fig:intro}, we give an informal summary of the present work. A \emph{measurement outcome} of a Boson Sampling experiment, defined by giving the interferometer and input and output states (see \cref{fig:intro}a) can be represented by a \emph{bipartite graph} encoding the connectivity of the interferometer (see \cref{fig:intro}b). Here, the vertices in the two partitions correspond to input, resp. output photons, and weighted edges correspond to nonzero transition amplitudes between individual photon positions. The probability amplitude of measuring the whole outcome of several indistinguishable photons is the sum over the perfect matchings in this graph, i.e. the permanent. In the shallow regime, the graph in \cref{fig:intro}b is sparse, and we capture this \emph{sparsity} using a \emph{tree decomposition} in \cref{fig:intro}c, where we group edges of the graph into nodes of a new tree structure, subject to simple connectivity conditions. The algorithm of Cifuentes and Parrilo exploits this tree decomposition to compute the permanent efficiently using a dynamical program~\cite{Cifuentes2015}.

Our main technical contribution
is shown in \Cref{fig:intro}d of the running example. We adapt the use of the \emph{Laplace expansion} of the permanent, a vital component of the algorithm by Clifford and Clifford~\cite{Cliffords2018}, to the permanent computation using tree decompositions.
Computing the probability of the $k$-th photon at a given output $r_k$ (``$\bullet^*$'' in \Cref{fig:intro}a), given a pattern of input and output modes for the previous $k-1$ photons (``$\bullet$''), is captured by the permanent $\per \W^{(r_k)}$ of a matrix
where only the last row depends on the choice of the new possible photon output $r_k$.
Clifford and Clifford~\cite{Cliffords2018} use the Laplace expansion:
\begin{equation}
  \label{eq:summary-permanent}
  \per \W^{(r_{k})}
  = \sum_{j = 1}^{k} \W^{(r_k)}_{k,j} \, \per \W^{(r_k)}_{\diamond (k,j)}.
\end{equation}
where $\W^{(r_k)}_{\diamond (k,j)}$ is the submatrix of $\W^{(r_k)}$ with row $k$ and column $j$ removed. As we remove the last row $k$, these matrices are independent of the choice of $r_k$. Clifford and Clifford exploit this fact to precompute the values of $\{ \per \W^{(r_k)}_{\diamond (k,j)} \}_j$ to later separately do a simple calculation of the $m$ possible output probabilities in polynomial time.

A naive approach to adapt the Cifuentes and Parrillo algorithm for nearest-neighbor shallow circuits would require a new tree decomposition for each of these submatrices, which degrades performance.
However, by ensuring that the nodes of the tree decomposition we use are in 1-to-1 correspondence to input photon vertices, and that the tree is linear (no branching, i.e. isomorphic to a path graph), we can compute each summand in the Laplace expansion by only a small modification of the original tree decomposition.
As showin in \cref{fig:intro}d, we formalize this as temporarily removing a node and replacing it by an ``empty'' dummy for each removed input photon, though the dummy may need some minimal contents. We may conceptualize this as a \emph{``machine head''} walking the tree from one end to the other, at each step replacing the node it visits, and computing the permanent $\per\W^{(r_k)}_{\diamond(k,j)}$ in (\ref{eq:summary-permanent}). These permanents together form the \emph{Laplace decomposition}, and this order of computing then allows us to reuse almost all dynamic programming tables required by the Cifuentes and Parrilo algorithm.

We remark that our abundant reuse of a single global tree decomposition requires us to work in the non-collision setting of $m = \Omega(n)$: a collision of output photons introduces a cycle in the relevant graph, and invalidates the global tree decomposition.

\subsection{Structure of the paper}

In \Cref{sec:bs-cc}, we give a short introduction to linear optics, Boson Sampling, and the Clifford and Clifford algorithms. We offer a more pedagogical overview in \Cref{sec:boson-sampling,sec:clifford}. In \Cref{sec:tree-decompositions}, we introduce the tree decomposition of matrices and the algorithm of Cifuentes and Parillo for computing permanents of sparse matrices. This is complemented by \Cref{sec:appendix-CP}.

The construction of our new algorithm takes place in \Cref{sec:main-results}, where we present some necessary constructions representing shallow photonic circuits, and in \Cref{sec:algorithm}, where we present the algorithm itself.
We conclude and discuss open problems in \Cref{sec:concl-future-work}, and we expand on some technical details in \Cref{sec:banded-matrices,sec:app:tree-decomp-manip}.

\section{\label{sec:bs-cc} Boson Sampling and Clifford \& Clifford algorithm}

To understand our new algorithm, we first review important concepts of single-photon Boson Sampling, and the Clifford and Clifford algorithms. A more pedagogical exposition and detailed derivations are found in \Cref{sec:boson-sampling,sec:clifford}.

\subsection{\label{sec:intro-bs} Short introduction to single-photon Boson Sampling}

In an \emph{ideal} Boson Sampling circuit, the input and output can be written in a basis of quantum states corresponding to the location of $n$ photons in $m$ modes. These can be represented as \emph{number (Fock) states} $\ket{\vec n} = \ket{n_1, \dots, n_m} \in (\C^n)^{\otimes m}$, where $n_i \in \N$ is the number of photons in mode~$i$, and $\sum_i n_i = n$. In this paper, however, it will be useful to see these states in the \emph{first quantization} representation where every photon corresponds to a \emph{qudit} of dimension $m$. Its logical information encodes the mode in which the photon is located, e.g. a basis state $\ket i$ means a photon in mode $i$. We ensure the indistinguishability of photonic states by writing the global quantum state as a symmetric superposition over all permutations of the qudits:
\begin{equation}
  \label{new:eq:qudit-representation}
  \ket{\Phi_{\vec n}} =
  \frac{1}{\sqrt{n! \prod_{i=1}^m n_i!}} \sum_{\sigma \in \Sg_n} \sigma \ket{\vec z},
\end{equation}
where we write $\ket{\vec z} \deq \ket1^{\otimes n_1} \otimes \ket2^{\otimes n_2} \otimes \cdots \otimes \ket{m}^{\otimes n_m}$ with the qudit factors sorted in non-descending order, where $\sigma$ is a permutation in the symmetric group on $n$ elements $\Sg_n$, and the prefactor on the RHS is its normalization~\cite{Moylett2018}. By convention, we denote the sorted product states as $\ket{\vec z}$, and various reorderings by permutations as $\ket{\vec r}$.

An \emph{interferometer} is a linear optical device that is represented by a unitary matrix $\U \in \mathbb{U}(m)$ acting linearly on a single-qudit basis state $\ket i$ as
\begin{equation}
  \ket{i} \xmapsto\U \sum_{j=1}^m \U_{i,j} \ket{j} \equiv \U \ket{i}.
\end{equation}
On a basis state of $n$ photons:
\begin{equation}
  \label{eq:U-action-on-z}
  \ket{\vec z} \equiv \bigotimes_{k=1}^n \ket{z_k}
  \xmapsto\U
  \bigotimes_{k=1}^n (\U \ket{z_k}) = \U^{\otimes n} \ket{\vec z},
\end{equation}
where we have introduced the notation $\ket{z_k}$ to mean the $k$-th qudit in the product $\ket{\vec z}$. Note unitarity of $\U$ ensures the total number of photons is preserved \cite{kok2010}.

\subsubsection{\label{sec:outc-prob} Outcome probabilities}

For a given input state $\ket{\vec n}$ the output probability distribution of a boson sampling circuit of interferometer $\U$ is given by the Born rule:
\begin{equation}
  \label{eq:bornrule}
\Pr[\vec n' | \vec n] = \left|\braoket{\Phi_{\vec n'}}{\:\U^{\otimes n}\:}{\Phi_{\vec n}}\right|^2
\end{equation}
where $\ket{\vec n'}$ is the output of the photon number measurement, and where
we have $\sum_i n'_i =\sum_i n_i = n$, as the photon number is preserved by ideal lossless circuits.
Rewriting (\ref{eq:bornrule}) in first quantization, we obtain
\begin{subequations}
  \begin{align}
    \Pr \left[ \vec{n}' \middle| \vec{n} \right]
    & = \label{eq:outcome-probability-permanent-expanded}
      \left|
      \frac{1}{\sqrt{\prod_{i=1}^m n_i! \; n'_i!}}
      \sum_{\pi \in \Sg_n} \;
      \prod_{j=1}^{n} \;
      \V^{\vec n', \vec n}_{j, \pi(j)}
      \right|^2 \\
    & = \label{eq:outcome-probability-permanent-general}
      \frac{\big| \per \V^{\vec n', \vec n} \big|^2}{\prod_{i=1}^m n_i! \, n'_i!},
  \end{align}
\end{subequations}
where we define a matrix $\V^{\vec n', \vec n}$ to represent the transition amplitudes of individual photons in the $\ket{\vec n} \to \ket{\vec n'}$ process, in the qudit basis. Its components are $\V^{\vec n', \vec n}_{l,h} \deq \braoket{z'_l}{\:\U\,}{z_h}$.
In (\ref{eq:outcome-probability-permanent-expanded}) we recognize the permanent of $\V^{\vec n', \vec n}$. We provide a complete derivation in \Cref{sec:boson-sampl-outcome-probabilities}.

\subsubsection{\label{new:sec:graph-repr}Graphical representation of $\V^{\vec n', \vec n}$}

\begin{figure}[t]
  \centering
  \subfloat[\label{new:fig:graphical-rep-U-to-V}]{%
    \includegraphics[scale=0.7]{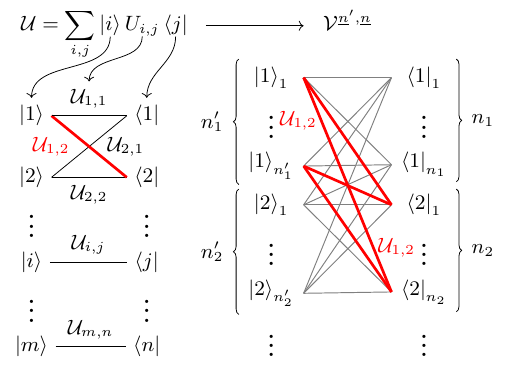}
  }
  \qquad
  \subfloat[\label{new:fig:graphical-rep-V}]{%
    \includegraphics[scale=0.7]{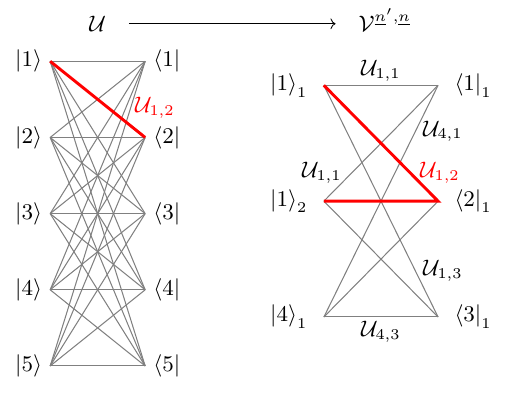}
  }
  \caption{Graphical representation:
    \textbf{(a)}~Graphs $G(\U)$ and $G(\V^{\vec n', \vec n})$ and their relationship. Left (right) partition corresponds to outputs (inputs). Note that $\U$ is the biadjacency matrix of $G(\U)$, and likewise $\V^{\vec n', \vec n}$ of $G(\V^{\vec n', \vec n})$. The output partition of the graph $G(\V^{\vec n', \vec n})$ consists of $n'_i$ copies $\neket i_1, \dots, \neket i_{n'_i}$ of output vertex $\ket{i}$ of $G(\U)$; inputs are similarly copied. If there is an edge $e=(\neket{i}, \nebra{j})$ in $G(\U)$, then $\smash{G(\V^{\vec n', \vec n})}$ has edges $(\ket{i}_a, \bra{j}_b)$ for all $a,b$, and these have the same weight as $e$: We show in \textbf{\color{red}red} the edge $(\ket{1}, \bra{2})$ of $G(\U)$ and its copies in $G(\V^{\vec n', \vec n})$, all with weights $\U_{1,2}$. Some edges are omitted for clarity.
    \textbf{(b)}~Example: $n = n' = 3$, $m = 5$, input $\ket{\vec{n}} = \ket{1,1,1,0,0}$ and output $\ket{\vec{n}'} = \ket{2,0,0,1,0}$.}
  \label{new:fig:graphical-rep}
\end{figure}

While $\U \in \C^{m \times m}$ represents an operator acting on modes, the matrix $\V^{\vec n', \vec n} \in \C^{n \times n}$ encodes the behaviour of the $n$ photons as they pass through the interferometer. In general, we get $\V^{\vec n', \vec n} \in \C^{n \times n}$ by taking $n'_i$ copies of row $i$ in $\U$, creating an $n \times m$ matrix, and from this $n_j$ copies of column $j$ \cite{Moylett2018}. In particular, unoccupied modes, or photon \emph{collisions}, lead to absent or copied rows/columns, respectively.

To facilitate understanding of the graphical methods used in \Cref{sec:tree-decompositions,sec:main-results,sec:algorithm}, we fruitfully interpret both $\U$ and $\V^{\vec n', \vec n}$ as weighted adjacency matrices of bipartite graphs, $G(\U)$ and $G(\V^{\vec n', \vec n})$ respectively, as shown in \Cref{new:fig:graphical-rep}, with a particular labeling of vertices useful to understand the transition from $\U$ to $\V^{\vec n', \vec n}$.
The graph $G(\U)$ encodes the relation from input to output modes of the linear interferometer. Vertices $\neket i$ ($\nebra j$) correspond to output (input) modes, and edges to non-zero probability input-output transitions.
The graph $G(\V^{\vec n', \vec n})$ has as vertices copies of $\neket i, \nebra j$ from $G(\U)$ labeled $\neket i_k, \nebra j_l$.
Edges are likewise copied from $G(\U)$, along with their weights, and they represent photons travelling between modes in the $\ket{\vec r}$ basis. The graph represents the possible \emph{paths} that all the \emph{distinguishable} photons may take from input $\ket{\vec z}$ to output $\ket{\vec z'}$. Each choice of which input photon becomes which output photon corresponds to a perfect matching in $G(\V^{\vec n', \vec n})$, and the amplitude of that matching is the product of amplitudes of its edges. The total amplitude of $\neket{\vec z} \to \neket{\vec z'}$ (representing $\neket{\vec n} \to \neket{\vec n'}$) is the sum over all matchings, giving the permanent $\per \V^{\vec n', \vec n}$. We show a simplified example of the relation to permanent in \cref{fig:per-graphs-matchings}.

\begin{figure}[t]
  \includegraphics[scale=0.9]{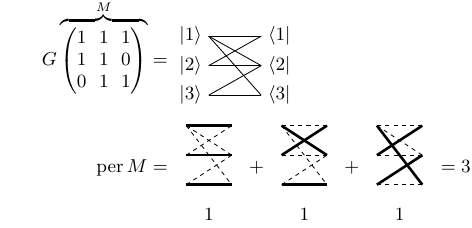}
  \vspace*{-0.7em}
  \caption{Relation between $G(M)$ and $\per M$ for an example binary matrix $M$: $\per M$ counts the perfect matchings of~$G(M)$.}
  \label{fig:per-graphs-matchings}
\end{figure}

\subsection{\label{sec:cliff-cliff-algor} Clifford and Clifford algorithms}

The work~\cite{Cliffords2018} by Clifford and Clifford presents three conceptual steps to arrive at the current best \emph{general} algorithm with running time $\cO(n 2^n)$. In the present summary, and in our pedagogical exposition in \Cref{sec:clifford}, we refer to these as Algorithms CC-A, CC-B, and \ref{alg:cliffords-c}. Note that~\cite{Cliffords2018} presents CC-B and CC-C together as a single algorithm. In this section we focus on the transition from CC-B to the faster CC-C which notably uses the crucial step of \emph{Laplace expansion} of the permanent, adapting which allows us to obtain our algorithm in \Cref{sec:algorithm}.

\subsubsection{\label{new:sec:algorithms-cca-ccb}Algorithms CC-A and CC-B}

The idea behind CC-A consists of sampling in the $\ket{\vec{r}}$ basis (see \Cref{sec:intro-bs}) using a chain rule procedure that builds the sample one photon at a time. We work as if the photons were distinguishable by assigning them labels, and we compute marginal probabilities of partial samples $(r_1, \dots, r_k)$. The trick is to create a sample of one photon, commit to it, then enlarge to two photons, etc.
Marginals of partial samples involve summing over the choices of inputs photons that are sampled, adding complexity. This is sped up in CC-B by permuting the input photons by a random uniformly distributed $\alpha \in \Sg_n$, \emph{new for each sample}. As shown in~\cite{Cliffords2018}, this gives the correct sampling distribution over several samples.

The marginals in CC-B require the computation of $\per \widetilde\V^{(r_1, \dots, r_k), \alpha([k])}$
where the matrix fulfils the same function as $\V^{\vec n', \vec n}$, but it represents a partial sample and is indexed in qudit basis.
Its components are $\widetilde\V^{(r_1, \dots, r_k), \alpha([k])}_{i,j} = \braoket{r_i}{\U}{\alpha(j)}$: observe that $\alpha([k]) = \{ \alpha(1), \dots, \alpha(k) \}$ selects the input modes and $(r_1,\dots,r_k)$ the outputs. We explain the steps to obtain CC-A and CC-B, as well as the running $\cO(mn2^{n+1})$ of CC-B, in \Cref{sec:clifford}.

\subsubsection{\label{sec:lapl-expans-algor}Laplace expansion: algorithm CC-C}

The final improvement of~\cite{Cliffords2018} transfers the factor $m$ from the dominant exponential term to a less relevant quadratic term, leading to an $\cO(n 2^n)$ algorithm that generates a sample with a factor of $2$ of the cost of computing a single output probability. This result is achieved using the \emph{Laplace expansion} of the permanent, defined below, and this is the idea we adapt to tree decomposition techniques in \Cref{sec:algorithm}. We summarize the algorithm CC-C as pseudocode in \Cref{fig:CC-C-float} in \Cref{app:laplace-expansion}.

When expanding the sample from $k-1$ to $k$ photons, we need to compute $m$ permanents of matrices $\widetilde\V^{(r_1, \dots, r_{k}), \alpha([k])}$ where all $k$ columns are shared (predetermined input photons), and all rows except $r_k$ are also the same (previously sampled output photons).
We simplify notation and denote $\W = \widetilde\V^{(r_1, \dots, r_{k-1}), \alpha([k])}$ the $(k-1) \times k$ matrix containing the shared rows and columns, and denote $\W^{(r_{k})} = \widetilde\V^{(r_1, \dots, r_{k}), \alpha([k])}$ the matrix obtained by adding the row $r_{k}$ of $\U$ to $\W$. Using the Laplace expansion of the permanent, we write:
\begin{subequations}
  \begin{align}
    \label{new:eq:laplace-expansion-per-general}
    \per \W^{(r_{k})}
    & = \sum_{j = 1}^{k} \W^{(r_{k})}_{k,j} \, \per \W^{(r_{k})}_{\diamond (k,j)} \\
    \label{new:eq:laplace-expansion-per}
    & = \sum_{j = 1}^{k} \U_{r_{k}, \alpha(j)} \, \per \W_{\diamond j}.
  \end{align}
\end{subequations}
In \eqref{new:eq:laplace-expansion-per-general}, $\W^{(r_{k})}_{\diamond(k,j)}$ is the $(k-1) \times (k-1)$ submatrix of $\W^{(r_{k})}$ with the row $k$ and column $j$ removed. By construction, the matrices $\W^{(r_{k})}_{\diamond(k, j)}$ are the same for all choices of $r_{k}$. In \eqref{new:eq:laplace-expansion-per}, we replace this by $\W_{\diamond j}$ denoting the $(k-1) \times (k-1)$ submatrix of $\W$ with column $j$ removed.

The final improvement by Clifford and Clifford is an adaptation of the Glynn's permanent algorithm~\cite{Glynn2010} to compute the family  $\{ \per \W_{\diamond j} \}_{j=1}^{k}$ in combined time $\cO(k 2^{k})$ and space $\cO(k)$. This is the same scaling, up to a constant, as that required to compute the permanent of the whole $k \times k$ matrix.
Once these permanents are known, computing the $m$ output probabilities requires only the use of the Laplace expansion~\eqref{new:eq:laplace-expansion-per}, taking $\cO(mk)$ operations. We elaborate on the time complexity in \Cref{app:CP-running-time}.

\section{\label{sec:tree-decompositions} Cifuentes and Parillo Algorithm}

We now present the concept of \emph{tree decompositions} of a graph and its use in the algorithm by Cifuentes and Parrilo (CP) to compute the permanent of sparse matrices using dynamic programming \cite{Cifuentes2015}. We will adapt this to CC-C in \Cref{sec:algorithm}.

\subsection{Tree Decompositions}

A tree decomposition is a structure that captures how (dis)connected a graph is, and it can under some conditions speed up the computation of solutions of instances of problems that are otherwise worst-case
NP-hard~\cite{Voigt2016}. Finding an optimal tree decomposition can itself be NP-hard, however, there exist good heuristics that allow improvements for practical problems. An example where tree decomposition is used by the quantum information community is the optimization of tensor network contractions. In our case, the Cifuentes and Parrilo algorithm exploits tree decompositions to compute permanents of structured matrices faster than using the standard algorithms by Ryser or Glynn.

Let $G$ be a bipartite graph of a matrix, in the sense of \Cref{new:sec:graph-repr}, where the set of vertices $V = \Rows \sqcup \Cols$ consists of row and column labels ($\Rows, \Cols$ respectively). The edges $E \subseteq \Rows \times \Cols$, each connecting one row and one column, correspond to the nonzero matrix elements.
In the following, we use the notation $2^X$ for the powerset of a set $X$, i.e. the set of all subsets of $X$.

\begin{figure}[t]
  \includegraphics[scale=0.9]{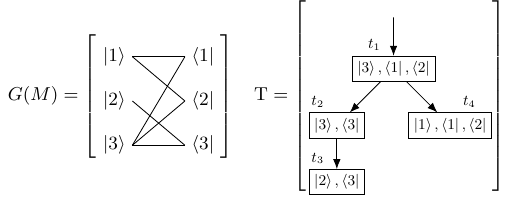}
  \caption{Example tree decomposition $\mathrm{T}$ of some graph of matrix $G(M)$. This tree is rooted at $t_1$.}
  \label{fig:example-treedec}
\end{figure}

\begin{definition}[tree decomposition of a bipartite graph]
  \label{def:treedec}
  Let $G$ be a bipartite graph. A tree decomposition $\mathrm{T} = (T, \rho, \kappa)$ of $G$ is a triple consisting of a tree $T$ rooted in $\rt(T)$, where for each \emph{node} $t \in T$, we denote $\ch(t)$ the unordered set of its children; and
  the functions $\rho : T \to 2^\Rows$ and $\kappa : T \to 2^\Cols$ that label the nodes with subsets of $\Rows$ and $\Cols$. We say that a node $t \in T$ \emph{contains} the elements of $\rho(t)$ and $\kappa(t)$.
  A tree decomposition must satisfy the following axioms:
  \begin{enumerate}[label={(\bfseries T\arabic*)}, ref={(T\arabic*)}]
  \item\label{def:treedec-axiom-vertex-cover}
    The union of contents over the entire tree covers all vertices of the graph:
    $\rho(T) = \Rows$ and $\kappa(T) = \Cols$, where we denote $\rho(T) = \bigcup_{t \in T} \rho(t)$, and analogouly $\kappa(T) = \bigcup_{t \in T} \kappa(t)$.
  \item\label{def:treedec-axiom-edge-cover}
    Every edge of the graph is contained in some node of the tree. This means that for each $(r, c) \in E$, there exists a node $t \in T$ that contains both its endpoints, i.e. $r \in \rho(t)$ and $c \in \kappa(t)$.
  \item\label{def:treedec-axiom-subtree}
    For every $r \in \Rows$ (resp. $c \in \Cols$), the set of all nodes that contain it is a subtree of $T$, i.e. all nodes containing $r$ (resp. $c$) are connected in $T$.
  \end{enumerate}
\end{definition}

We show an example tree decomposition in \Cref{fig:example-treedec}.
By convention, we write the decomposition upright as $\mathrm{T}$, and the underlying  tree slanted as $T$. Whenever we say the word \emph{node}, we mean some $t \in T$ in a  tree decomposition, and when we say \emph{vertex}, it means some $v \in V$ of the graph $G$ being studied.
We write $T_t$ for the \emph{subtree} of $T$ rooted at $t \in T$, that is the tree that contains $t$ and all its descendants. An important consequence of Axiom~\ref{def:treedec-axiom-subtree} that will be relevant later is the following:
if a vertex or an edge is contained in both $T_{c}$ and $T_{c'}$, where $c, c'$ are two distinct children of $t$, then it must also be contained in the parent node $t$.

\subsubsection{Treewidth}

For any graph $G$, there are generally many tree decompositions: we denote this collection $\Tds(G)$. However, not all of them are useful, e.g.~the trivial case of a single node containing the entire graph $G$. The following definition captures how good a tree decomposition is:

\begin{definition}[treewidth~\cite{Cifuentes2015,Voigt2016}]
  \label{def:treewidth}
  For a tree decomposition $\mathrm{T} = (T, \rho, \kappa)$, we define its \emph{treewidth} as:
  \begin{subequations}
    \begin{equation}
      \label{eq:treewidth-of-treedec}
      \tw(\mathrm{T}) \deq \max_{t \in T} \left\{ \vphantom{\big|} |\rho(t)| + |\kappa(t)| \right\} - 1,
    \end{equation}
    that is the size of the \emph{largest} node $t$ in terms of its contents, minus one. For a graph $G$, we define its treewidth as:
    \begin{equation}
      \label{eq:treewidth-of-graph}
      \twG(G) \deq \min_{\mathrm{T} \in \Tds(G)} \tw(\mathrm{T}),
    \end{equation}
  \end{subequations}
  that is the minimum treewidth of any possible tree decomposition of $G$. We distinguish the treewidth of a graph with a hat.
\end{definition}

The treewidth measures how different a graph is from a tree, and we define the RHS of~\eqref{eq:treewidth-of-treedec} with $-1$ so that a tree $Q$ has $\twG(Q) = 1$.

Both in tensor network contraction and in our permanent computations the running time and memory requirements scale exponentially with the treewidth of the graph. Computing with tree-like graphs is in general much easier than graphs with higher treewidth which contain many cycles~\cite{Voigt2016}. In our case, the intuition is that a graph with high treewidth, and thus many cycles, allows many perfect matchings, i.e. summands of the permanent, which leads to a more expensive computation.

\subsection{\label{sec:permanent-algorithm}Permanent algorithm}

\begin{figure}[tb]
  \begin{algorithm}[H]
    \SetAlgoRefName{CP}
    \caption{\textsf{Bipartite permanent}}
    \label{alg:cifuentes}
    \KwIn{
      \begin{itemize}
      \item a matrix $M \in \C^{n \times n}$,
      \item a tree decomposition $\mathrm{T} = (T, \rho, \kappa)$ of $G(M)$
      \end{itemize}
    }
    \KwResult{the permanent $\per M$}

    $\mathit{order} \gets$ list of nodes in $T$, in topological order, starting from leaves

    \For{$t \in \mathit{order}$}{

      $Q[t](R, C) \gets \per \restr{M}{R,C}$ for $R \subseteq \rho(t)$ and $C \subseteq \kappa(t)$

      \If{$t$ is a leaf, i.e. $\ch(t) = \varnothing$}{
        $P[t] \gets Q[t]$
      }
      \Else{

        \For{$c_j \in \ch(t)$}{
          $\begin{aligned}
            &
              \begin{cases}
                \Delta^\rho_{c_j} & \gets \rho(c_j) \setminus \rho(t) \\
                \Lambda^\rho_{c_j} & \gets \rho(c_j) \cap \rho(t)
              \end{cases}
            &
            &
              \begin{cases}
                \Delta^\kappa_{c_j} & \gets \kappa(c_j) \setminus \kappa(t) \\
                \Lambda^\kappa_{c_j} & \gets \kappa(c_j) \cap \kappa(t)
              \end{cases}
          \end{aligned}$

          \For{$R \subseteq \Lambda^\rho_j$ and $C \subseteq \Lambda^\kappa_j$}{
            $Q'[t \,|\, c_j](R, C) \gets (-1)^{|R|} \cdot Q[t](R,C)$

            $Q''[t \gets c_j](R, C) \gets P[c_j](R \cup \Delta^\rho_{c_j}, C \cup \Delta^\kappa_{c_j})$
          }
        }

        \For{$R \subseteq \rho(t)$ and $C \subseteq \kappa(t)$}{
          $P[t](R,C) \gets$ subset convolution of $Q[t]$ and all $Q'[t|c_j], Q''[t \gets c_j]$ as in (\ref{new:eq:CP-subset-convolution})
        }
      }
    }

    \KwRet{$P[r](\rho(r), \kappa(r)) = \per M$} where $r = \rt(T)$
  \end{algorithm}
  \caption{Pseudocode of the Cifuentes and Parrilo algorithm.}
  \label{fig:pseudocode-CP}
\end{figure}

We now present the permanent algorithm by Cifuntes and Parillo from~\cite[Algorithm 2]{Cifuentes2015}, to which we refer as Algorithm \ref{alg:cifuentes}, and whose pseudocode can be found in \cref{fig:pseudocode-CP}. The algorithm exploits the structure and sparsity of the matrix to accelerate the computation of its permanent. Using a tree decomposition of the graph of the matrix, it is able to decompose the matrix into smaller parts and compute their permanents. The tree structure then tells it how to join these parts together. It is a dynamic programming algorithm that stores values in tables for later use. For an $n \times n$ matrix and a tree decomposition of its graph of treewidth $\omega$, the running time of \ref*{alg:cifuentes} is $\cO(n 2^\omega \omega^2)$.\footnote{ The running time is conventionally given as $\tcO(n 2^\omega)$, which hides polynomial factors in $\omega$, but the factor $\omega^2$ can be recovered from~\cite[Lemma 2 and proof of Lemma 14]{Cifuentes2015} as detailed in \Cref{app:CP-running-time}.}

\subsubsection{Dynamic programming tables}

The Algorithm~CP builds the permanent of matrix $M$ from leaves upward. For all nodes and for the steps between child and parent nodes, we need to build two dynamic programming tables $Q[t](R,C)$ and $P[t](R,C)$, both indexed by a subset of rows $R \subseteq \rho(t)$ and of columns $C \subseteq \kappa(t)$ belonging to the node $t$. The table $Q$ stores the permanents of submatrices with rows and columns belonging to the node $t$,
while the table $P$ allows to combine that information with the permanents of its children using a subset convolution:
\begin{equation}
  \label{new:eq:CP-subset-convolution}
  \begin{aligned}
    P[t](R,C) = & \sum \bigg( Q[t](R_t, C_t) \\
                & \times \prod{}_{c_j \in \ch(t)} \ {}
                  Q'[t | c_j](R'_{c_j}, C'_{c_j}) \\
                & \hphantom{\times \prod{}_{c_j \in \ch(t)} \ {}} \times
                  Q''[t \gets c_j](R''_{c_j}, C''_{c_j}) \bigg).
  \end{aligned}
\end{equation}
Note that in the computation of $P[t]$, we also use two \emph{helper tables} $Q'[t|c_j]$ and $Q''[t \gets c_j]$ for each child $c_j$ of $t$. These latter bring the values of $P[c_j]$ into the computation of $P[t]$. The summation runs over partitions of $R$ (resp. $C$) into a families of subsets $\{ R_i\smash{^{(\prime,\prime\prime)}} \}_i$ (resp. $\{ C_i\smash{^{(\prime,\prime\prime)}} \}_i$). We expand on the computation of these tables and valid partitions in \Cref{sec:appendix-CP:tables}. Note that the final result, the permanent of $M$, is found in the root node:
\begin{equation}
  P[r](\rho(r), \kappa(r)) = \per M.
\end{equation}

\section{\label{sec:main-results} Preliminaries to new algorithm}

Before we can construct our main contribution, an algorithm for polytime classical simulation of Boson Sampling from shallow interferometers, we need to review and define a few necessary constructions.

\subsection{Nearest-neighbour circuits}
\label{sec:near-neighb-circ}

We are interested in architectures of ideal intereferometric circuits
composed of nearest-neighbour interactions between modes.
It is well known that such interferometers can be constructed from phase-shifters and beamsplitters~\cite{Carolan2015,Bell2021} which are, respectively, one- and two-mode gates. The latter is defined as the following unitary acting on modes (cf. $\U$)~\cite{Campos1989}:
\begin{equation}
  \label{eq:beamsplitter-matrix}
  B(\theta, \phi_T, \phi_R) =
  \begin{pmatrix}
    e^{i \phi_T} \cos \theta & e^{i \phi_R} \sin \theta \\
    -e^{-i \phi_R} \sin \theta & e^{-i \phi_T} \cos \theta
  \end{pmatrix},
\end{equation}
where $\theta$ is the coupling between the two modes, and $\phi_T, \phi_R$ are the transmitted, resp. reflected, phase shift. Similarly to~\cite{Jozsa2006}, we absorb the phase-shifters, being single-mode operations, into adjacent beamsplitters whenever possible; thus our circuits consist only of beamsplitters.

\begin{figure*}[t]
  \centering
  \hspace*{-2cm}\includegraphics[scale=0.85]{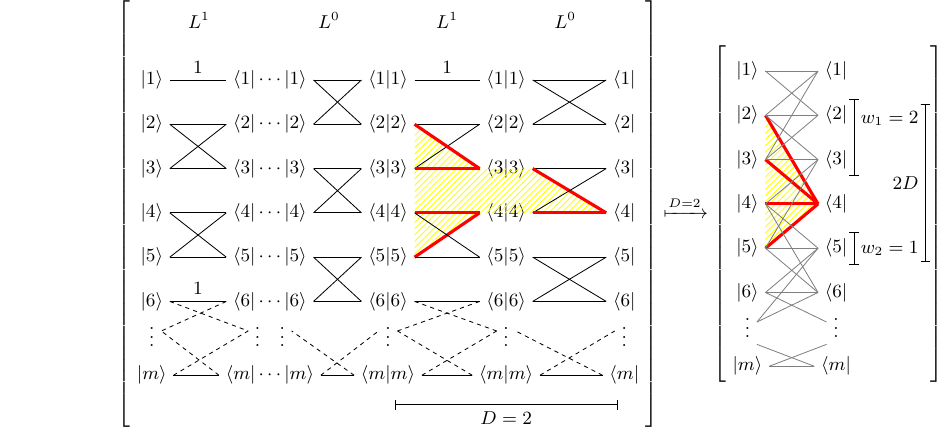}
  \caption[Architecture]{An even-depth example (last layer $L^1$) of an alternating array of beamsplitters in our architecture, represented using the graphical notation from \cref{new:fig:graphical-rep}. The columns represent the layers of gates $L^0$ and $L^1$, and the crossings are beamsplitters. For each mode, one beamsplitter connects it with the mode above, and the next beamsplitter with the mode below, or vice-versa. The bold red lines show the possible paths of photons starting in mode $4$, corresponding to its \emph{causal cone} (filled yellow), over two layers of the circuit. The RHS is the corresponding graph $G(\U)$ resulting from the first two layers of the circuit.}
\label{fig:ABA-schematic}
\end{figure*}

We consider a construction with the beamsplitters arranged in alternating layers as displayed in Figure~\ref{fig:ABA-schematic}, as in the Clements rectangular decomposition~\cite{Clements2016} that guarantees the universality of the interferometric circuit while minimizing its depth $D$. This decomposition also has practical advantages, as it guarantees more uniform losses among the different output modes which limits the effects of those imperfections. A full circuit is composed as an alternation of two slightly different types of layers of beamsplitters, $L^0$ and $L^1$ (where we omit parameters), that allow the spreading of the correlations over distant modes following a causal cone, emphasized in \cref{fig:ABA-schematic}.

As consequence of the causal cone, the unitary $\U$ of the circuit has a banded structure, meaning there exist numbers $w_1, w_2$ s.t. $\U_{i,j} = 0$ for $i - j > w_1$ or $j - i > w_2$. Therefore, the \emph{bandwidth} $w$ is bounded as:
\begin{equation}
  \label{new:eq:bandwidths-of-ABA}
  w = w_1 + w_2 + 1 \le 2D,
\end{equation}
as shown by Lemma~\ref{lemma:bandwidths-of-ABA} in \Cref{sec:banded-matrices}.
This can be intuitively seen in \cref{fig:ABA-schematic} where the bandwidth corresponds to the furthest distance (in modes) a photon can travel away from its input mode. It is easy to see that each beamsplitter layers expands each of $w_1$ and $w_2$ by at most $1$, and the width after the first layer is $2$.

\subsection{\label{sec:tree-decomp-band}Tree decomposition of a banded matrix}

There are different tree decomposition one can design for an $m \times m$ banded matrix $\U$. For algorithmic design and simplicity we will use a tree decomposition $\mathrm{T^C} = (T, \rho, \kappa)$ shown in \Cref{fig:TC-on-matrix} that is not optimal but simple.\footnote{The optimal decomposition uses two nodes per column, each containing $w-1$ rows, where $w$ is the number of rows with a nonzero in a column. These nodes are neighbours, one containing the upper $w$ rows, and the other the lower. This was given by Cifuentes and Parrilo~\cite{Cifuentes2015}, modulo exchange of rows and columns between their paper and ours.}
The decomposition $\mathrm{T^C}$ consists of a linear tree $T=\{t_1, \dots, t_m\}$ where the node $t_i$ has a single child $t_{i+1}$ for $i=1,\dots, m-1$ (and $t_m$ is a leaf). There is a correspondence between the nodes and columns of the matrix, such that $\kappa(t_i) = \{ \bra i \}$ for all $i \in [m]$. This tree decomposition associates to each node $t_i$, corresponding to column $i$, all rows $j$ where the matrix elements of that column may be nonzero: $\ket j \in \rho(t_i)$ iff $\U_{ji}$ is within the band. If $\U$ has bandwidths $(w_1, w_2)$, then in general:
\begin{equation}
  \rho(t_i) = \{ \ket{j - w_1}, \ket{j - w_1 + 1}, \dots, \ket{j + w_2} \},
\end{equation}
modulo non-existant row labels.
There are at most $w$ elements in $\rho(t_i)$, and together with a single element in $\kappa(t_i)$, recalling that the treewidth is defined with a constant $-1$ (see \eqref{eq:treewidth-of-treedec}), this makes the treewidth of the decomposition $\tw(\mathrm{T^C}) \le w \le 2D$.

\begin{figure}
  \includegraphics[scale=0.7]{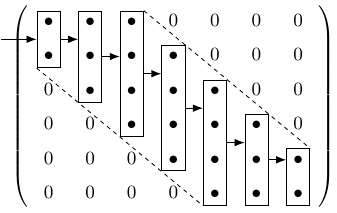}
  \caption{Example linear tree decomposition $\mathrm{T^C}$ where tree nodes are shown as boxes encompassing possibly nonzero elements (``$\bullet$''). Remark that nodes correspond to parts of columns of the matrix. The band is shown by dashed lines.}
  \label{fig:TC-on-matrix}
\end{figure}

\subsection{\label{new:sec:modif-orig-tree}Tree decompositions of submatrices with column permutations}

Having a low-width tree decomposition, we are able to compute the permanent of a matrix efficiently using Algorithm~\ref{alg:cifuentes}. However, while the original matrix $\U$ has a banded structure allowing us to find the decomposition, this structure is concealed when we permute columns and take submatrices (delete rows and columns) to obtain $\widetilde\V^{(r_1, \dots, r_k),\alpha([k])}$, transformations required by the Clifford and Clifford sampling algorithms. A priori, we need a new tree decomposition for each submatrix or column permutation, and these are no longer easy to find. However, in what follows we present a \emph{systematic approach} to obtain all tree decompositions of the required submatrices from the initial tree decomposition of~$\U$ by performing a set of \emph{simple and efficient manipulations}.

\subsubsection{\label{sec:permutation-columns} Permutation of columns}

An important step is to reorder the (first $n$) columns of the matrix $\U$ by a permutation $\alpha$. This acts as a graph isomorphism on~$G(\U)$:
\begin{equation}
  \label{eq:permutation-column-on-ket}
  \ket{i} \xmapsto\alpha \ket{\alpha^{-1}(i)}.
\end{equation}
This relabels columns without changing the connectivity structure, so we implement it on a tree decomposition again by just relabeling columns as in~(\ref{eq:permutation-column-on-ket}). The rest of the tree structure \emph{stays the same}.
Note that we conventionally act by $\alpha^{-1}$ within the column label, in order to address the matrix elements as $\U_{i,\alpha(j)}$.
Formal details can be found in \Cref{sec:app:permutations-columns}.

\subsubsection{\label{new:sec:submatrices} Submatrices}

In all of the Clifford and Clifford algorithms, we abundantly need to compute permanents of various submatrices. We need a way to exploit the known tree decomposition of the original matrix.

For any tree decomposition $\mathrm{T}$ of the matrix $\U$, and subsets $\Rows' \subseteq \Rows$ of rows and $\Cols' \subseteq \Cols$ of columns, selecting a submatrix $\restr{\U}{\Rows', \Cols'}$, we construct the \emph{restriction} of $\mathrm{T}$, denoted $\restr{\mathrm{T}}{\Rows', \Cols'}$, as the tree decomposition built on the same tree, where we simply delete all row and column labels not present in $\Rows'$ and $\Cols'$. We show an example in \Cref{fig:skipping-nodes-original,fig:skipping-nodes-restriction}. As seen in \Cref{sec:appendix-restr-tree-decomp}, this is a well-defined tree decomposition that indeed corresponds to the submatrix $\restr\U{\Rows', \Cols'}$. Importantly, this cannot increase the treewidth, bounded (in Lemma~\ref{lemma:treewidth-of-restriction}) as:
\begin{equation}
  \label{new:eq:width-of-decomposition}
  \tw\bigl( \restr{\mathrm{T}}{\Rows', \Cols'} \bigr)
  \le \min\bigl( \tw(\mathrm{T}), |\Rows'| + |\Cols'| - 1 \bigr).
\end{equation}

\subsubsection{\label{new:sec:redundancy-nodes}Redundancy of nodes}

\begin{figure}[t]
  \centering
  \subfloat[\label{fig:skipping-nodes-original}]{\includegraphics[scale=1]{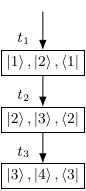}}
  \qquad
  \subfloat[\label{fig:skipping-nodes-restriction}]{\includegraphics[scale=1]{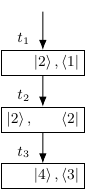}}
  \qquad
  \subfloat[\label{fig:skipping-nodes-redundant}]{\includegraphics[scale=1]{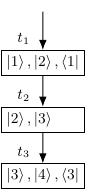}}
  \qquad
  \subfloat[\label{fig:skipping-nodes-skipped}]{\includegraphics[scale=1]{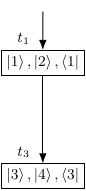}}
  \caption{Examples of tree restrictions: \textbf{(a)}~A decomposition $\mathrm{T} = \mathrm{T^C}$.
    \textbf{(b)}~A~restriction $\restr{\mathrm{T}}{\Rows', \Cols'}$ to rows $\Rows' = \{ \ket 2, \ket 4 \}$ and columns $\Cols' = \{ \bra 1, \bra 2, \bra 3 \}$. Vertices $\ket 1$ and $\ket 3$ are removed.
    \textbf{(c)}~A different restriction, this time only removing the column $\bra 2$, making the node $t_2$ redundant. It cannot contribute any new elements from the matrix. Instead, node $t_2$ is only used to make some values from table $P[t_3]$ available to $t_1$. We can thus remove $t_2$ and connect $t_1$ directly to $t_3$ to obtain the equivalent decomposition $\mathrm{T} \setminus t_2$ in~\textbf{(d)}.}
  \label{fig:skipping-nodes}
\end{figure}

Each node of our tree decomposition $\mathrm{T^C}$ contains a single column label. When restricting the decomposition, we may end up with a node $t \in T$ without any columns, i.e. $\kappa'(t) = \varnothing$. Indeed, our algorithm will perform such restrictions extensively.
In this case, this node contains no information (corresponds to a column no longer present), and we can remove it, obtaining a new decomposition $\mathrm{T} \setminus t$ where we directly connects the neighbours of $t$ (unless it is root or leaf), see \cref{fig:skipping-nodes-redundant,fig:skipping-nodes-skipped}. We expand on the formal details in \Cref{sec:redundancy-of-nodes}.

\section{\label{sec:algorithm}Algorithm for shallow circuits}

Given the above methods, we could drop-in replace all permanent computation in CC-C by Algorithm \ref{alg:cifuentes} to obtain a running time of $\cO(mn^22^\omega \omega^2)$. A similar method was used before in~\cite{Oh2021}. In what follows, however, we show how to achieve $\cO(n^22^\omega \omega^2) + \cO(\omega n^3)$ by finding an efficient way of computing the $k+1$ permanents of submatrices $\W_{\diamond j}$ from (\ref{new:eq:laplace-expansion-per}) using a single tree decomposition, leading to a scaling linear in $k$, reproducing the saving obtained in the original proposal of Clifford and Clifford~\cite{Cliffords2018}. The \emph{key tool} is to compute the $k+1$ permanents of submatrices $\W_{\diamond j}$ \emph{in the right order} by doing \emph{local updates} to the single initial (restricted) tree decomposition. This will allow us to \emph{exploit precomputed tables}, reducing the amount of required computation at each step, closely mimicking the idea of CC-C. Note that in order to reuse a global tree decomposition and tables computed on it, we require the non-collision setting of $m = \Omega(n^2)$.\footnote{This is because a collision leads to copies of an output vertex in $G(\widetilde\V^{(r_1,\dots,r_k), \alpha([k])})$, along with copies of edges. Generally, this introduces new cycles, invalidating the global tree decomposition.}

\subsection{\label{sec:preparation}Preparation}

We first prepare the main tree decomposition \mbox{$\mathrm{T} = \mathrm{T^C}$} of the graph $G(\U)$ of the banded unitary $\U$. Recall that each tree node contains a single column label (input of the circuit) and all row labels (outputs of the circuit) for which that column may have a non-zero element in~$\U$; see \Cref{sec:tree-decomp-band}.

We have a fixed pattern of input photons to the circuit $\ket{\vec n}$ with $n_i \in \{ 0, 1 \}$ for each input mode $i$. We restrict the tree decomposition $\mathrm{T}$ by removing the columns corresponding to unoccupied input modes, as well as the corresponding tree nodes; see \Cref{new:sec:submatrices,new:sec:redundancy-nodes}. We also remove those rows that only contain zeros after the deletion of columns.
The new decomposition $\tilde{\mathrm{T}}$ corresponds to the submatrix of $\U$, denoted $\V^{\vec n}$, restricted to the columns corresponding to occupied inputs.
The matrix $\V^{\vec n}$ may apparently have different band structure, but this is not a problem: the original structure is remembered by the restricted tree-decomposition whose treewidth remains at most $\omega = \tw \mathrm{T}$ by (\ref{new:eq:width-of-decomposition}).

\subsubsection{\label{sec:permutation-of-inputs}Permutation of inputs}

Next, we generate a uniformly random permutation $\alpha \in \Sg_n$ that we use to permute the columns by relabeling $\bra i$ to $\nebra{\alpha^{-1}(i)}$ in $\tilde{\mathrm{T}}$ as shown in (\ref{eq:permutation-column-on-ket}). This preserves all tree decomposition structure while ensuring that the new decomposition correctly describes the permuted matrix.
Note that after the permutation, the column label $\bra i$ is found in node $t_{\alpha(i)}$ of the tree.

\subsubsection{\label{sec:prec-q-tabl}Precomputing $Q$-tables}

We compute the $Q$ tables of the entire tree for later use. As each table $Q[t]$ is local to the node $t$, each containing a single column, this consists of permanents of $1 \times 1$ matrices, i.e. copies of components of $\V^{\vec n}$.\footnote{We need these tables for consistency with the formulation of the Cifuentes and Parrilo algorithm, though in practice, one may choose to refer directly to the matrix.}

\subsection{\label{sec:sampling-algorithm} The sampling algorithm}

Now we perform the sampling itself. We show first the special case of the first marginal and then generalize, showing that not only can we efficiently compute the permanents needed for the Laplace expansion, but by doing our tree manipulations in the right order, we achieve the ability to \emph{reuse almost all dynamic programming tables} for several calculations.

\subsubsection{\label{sec:first-marginal}First marginal}

The first marginal can be seen as expanding the empty sample to $k = 1$ photon. The Laplace expansion is trivial as the matrix
$\W_{\diamond 1}$ is empty. Equivalently, we compute the permanents of $1 \times 1$ matrices. The probability of the first photon being in output mode~$i$ is $p_1(i) = |\W^{(i)}_{1,1}|^2$, where the subscript on $p_1$ indicates that this is the first marginal. This probability is the modulus squared of:
\begin{equation}
  \label{eq:first-marginal-W}
  \W^{(i)}_{1,1} = \widetilde\V^{(i), (\alpha(1))}_{1,1} = \U_{i,\alpha(1)},
\end{equation}
where we note the use of notation referring to the original matrices (without column permutation), and write $\alpha$ explicitly for clarity.
The matrix element $\U_{i,\alpha(1)}$ is represented by the node $t_{\alpha(1)}$ containing the label $\bra 1$. We remark that one does not need to explore all $m$ output modes, as at most $\omega$ (contained in $\rho(t_{\alpha(1)})$) are connected to the input mode $\alpha(1)$ due to the circuit depth. We select one output mode at random, using the computed marginal probability distribution $p_1$, and store it in the register~$r_1$.

\subsubsection{\label{sec:generic-marginal}The $k$-th marginal}

The $k$-th marginal, for $k>1$, adds the $k$-th input photon to the sample of $k-1$ \emph{distinct} photon modes $(r_1, \dots, r_{k-1})$ which we have already obtained in previous steps. The latter determine the matrix $\W = \widetilde\V^{(r_1, \dots, r_{k-1}), \alpha([k])}$ made of the restriction of $\V^{\vec n}$ to rows $(r_1, \dots, r_{k-1})$, and to columns $\alpha([k])$.
The $k$-th marginal adds one more output photon, leading to $k \times k$ matrices $\W^{(r_{k})}$, for all $r_{k} \in [m] \setminus \{ r_1, \dots, r_{k-1} \}$. Recall that $\W^{(r_{k})}$ is the matrix obtained from $\W$ by adding a row $r_{k}$ of $\V^{\vec n}$, restricted to the same columns.

\begin{figure}[tb!]
  \centering
  \subfloat[\label{fig:tree-rotations:original}]{\includegraphics[scale=1]{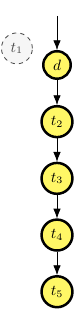}}
  \qquad
  \subfloat[\label{fig:tree-rotations:R1}]{\includegraphics[scale=1]{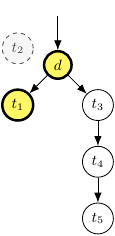}}
  \qquad
  \subfloat[\label{fig:tree-rotations:R2}]{\includegraphics[scale=1]{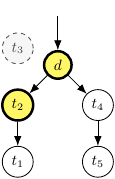}}
  \qquad
  \subfloat[\label{fig:tree-rotations:R3}]{\includegraphics[scale=1]{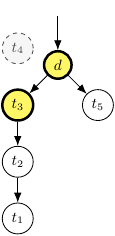}}
  \qquad
  \subfloat[\label{fig:tree-rotations:R4}]{\includegraphics[scale=1]{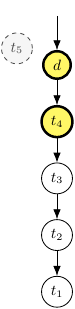}}
  \definecolor{goldenrod}{rgb}{0.85, 0.65, 0.13}
  \caption{Computation of $\per\W^{(r_{5})}$ using the Laplace expansion:
    \captionitem{a}~We start with $\mathrm{T}^\W_{\diamond 1}$ corresponding to~$\per \W_{\diamond \alpha^{-1}(1)}$. To compute this, replace the node $t_1$ by a dummy $d$ and run algorithm~\ref{alg:cifuentes}. The $P$ and $Q$ tables of all nodes need to be computed, which is shown by \textbf{\color{goldenrod}yellow} filling and \textbf{bold} outline.
    \captionitem{b}~Next, put $t_1$ back, and replace $t_2$ by a new dummy $d$ that is now the parent of both $t_1$, and of the subtree of $t_3$, obtaining $\mathrm{T}^\W_{\diamond2}$ needed to compute~$\per \W_{\diamond \alpha^{-1}(2)}$. All $Q[t_1], \dots, Q[t_5]$ tables have been computed and remembered in step (a), so we do not recompute them. Observe that each of $t_3, \dots, t_5$ has the same child (if any) as in step (a), so their $P$ tables do not need to be recomputed either. Marked in yellow are $t_1$ and $d$, the \emph{only} nodes for which we need to compute \emph{new} tables: only $P[t_1]$ for $t_1$, and both $Q[d]$ and $P[d]$ for~$d$.
    \captionitem{c,~d,~e}~The following steps, corresponding to $j=3,4,5$, respectively. In each, we mark yellow the nodes for which we need to compute the dynamic programming tables. At step $j$, the tables computed are $P[t_{j-1}]$, $Q[d]$ and $P[d]$. The latter contains $\per \W_{\diamond \alpha^{-1}(j)}$ in $P[d](\rho'(d), \varnothing)$.
    Observe for example in (d) that nodes $t_1$ and $t_2$ are not yellow: following the same idea as before, their children (and hence $P$ tables) have not changed from the previous step~(c).
  }\label{fig:tree-rotations}
\end{figure}

The matrix $\W$ is represented by the restriction $\mathrm{T}^\W$ of $\mathrm{\tilde T}$ that removes all column labels other than $\bra{1}, \dots, \bra{k}$ which are found in tree nodes $t_{\alpha(1)}, \dots, t_{\alpha(k)}$, respectively. We also need to remove all output labels except $\ket{r_1}, \dots, \ket{r_{k-1}}$.
Note that all nodes except for $t_{\alpha(1)}, \dots, t_{\alpha(k)}$ are redundant in this restriction, so we can skip them in the permanent computations. Equivalently, we delete the redundant nodes while keeping the remaining nodes connected, as in \cref{fig:skipping-nodes-skipped}; see \cref{new:sec:redundancy-nodes}.

The (unnormalized) marginal probability distribution $p_k$ of the $k$-th photon has values $p_k(r_k) = |\per\W^{(r_k)}|^2$. We compute $\per\W^{(r_k)}$ using the Laplace expansion from eq.~(\ref{new:eq:laplace-expansion-per}):
\[
  \per \W^{(r_k)} = \sum_{j = 1}^{k} \U_{r_{k}, \alpha(j)} \, \per \W_{\diamond j}.
\]
We compute the collection $\{ \per \W_{\diamond j} \}_{j=1}^k$ of permanents of the shared submatrices using a collection of closely related tree restrictions $\{ \mathrm{T}^\W_{\diamond j} \}_{j=1}^k$. Each $\mathrm{T}^\W_{\diamond j}$ is the restriction of $\tilde{\mathrm{T}}$, the global decomposition of $\V^{\vec n}$, to rows $\{ r_1, \dots, r_{k-1} \}$ and columns $\alpha([k]) \setminus \{ j \}$, that is, $\mathrm{T}^\W_{\diamond j}$ is the tree decomposition representing the matrix~$\W_{\diamond j}$.

\subsubsection{\label{sec:moving-head}Moving machine head: local manipulations $\mathrm{T}^\W_{\diamond j}$}

We could already compute the Laplace expansion using $\{ \mathrm{T}^\W_{\diamond j} \}_{j=1}^k$, evaluating them in some arbitrary order. However, if we view this as a sequence of local manipulations to the overarching tree decomposition $\mathrm{T}^\W$  corresponding to $\W$, and if we do these in an appropriate order, we unlock a further advantage of reusing the dynamic programming tables between the individual permanents.

The idea, as shown in an example in \Cref{fig:tree-rotations} where each subfigure corresponds to one step, is to imagine a machine head walking sequentially along the nodes of $\mathrm{T}^\W$ from one end to the other, in analogy to a Turing machine walking the tape. At each step, it temporarily replaces the node where the head is, say $t$, by a \emph{dummy} node $d$ that contains no new information, and sets this dummy $d$ to be the root of the tree. Then it computes the permanent from this tree using the algorithm CP. Finally, it restores the original node $t$ and moves on to the next one, repeating the procedure.
We elaborate on the details in three conceptual parts:

\paragraph{Order of manipulations}
Let $(j_1, \dots, j_k)$ be some permutation of $[k]$ that determines the order in which we compute the permanents, i.e. we start with $\per \W_{\diamond j_1}$, continue with $\per \W_{\diamond j_2}$, etc. To minimize the set of $P$-tables changing between steps, choose the values $j_1,\dots, j_k$ so that we manipulate the nodes $t_{\alpha(j_1)}, \dots, t_{\alpha(j_k)}$ in the order as they appear in the tree of $\mathrm{T}^\W$, i.e. $t_{\alpha(j_\ell)} = t_\ell$. Thus $j_\ell = \alpha^{-1}(\ell)$ for all~$\ell$.\footnote{We may also choose to walk the tree in the opposite direction. This does not make a difference.}
Note that to simplify notation in what follows and in \cref{fig:tree-rotations}, we number the nodes of $\mathrm{T}^\W$ consecutively as  $t_1, t_2, \dots, t_k$, ignoring deleted nodes from $\tilde{\mathrm{T}}$ where the numbering (but not order!) may differ.

\paragraph{Node replacement}
To compute $\per \W_{\diamond j}$, we need to remove the column $\bra{j}$ and thus node $t_{\alpha(j)}$ of $\mathrm{T}^\W$, obtaining $\mathrm{T}^\W_{\diamond j}$. We will need to consider the part of the tree to the left and right of $t_{\alpha(j)}$, so conceptually it is useful to create a \emph{dummy} node containing no information instead of connecting the neighbours of $t_{\alpha(j)}$ directly. To satisfy Axiom~\ref{def:treedec-axiom-subtree}, $\rho'(d) = \rho'(t_a) \cap \rho'(t_b)$ if $d$ has neighbours $t_a$ and $t_b$, and $\rho'(d) = \varnothing$ if $d$ is at either end of the tree, where $\rho'$ is the row function of $\mathrm{T}^W_{\diamond j}$. It contains no columns. After computing $\per \W_{\diamond j}$ we return the node $t_{\alpha(j)}$, and move on to the next node with a \emph{new} dummy.

From now on, it will be useful to refer to nodes by their order in the tree as $t_j$, instead of $t_{\alpha(j)}$ used above to more easily address matrix columns. Observe that $t_j$ corresponds to column $\nebra{\alpha^{-1}(j)}$ and removing it allows computing $\per \W_{\diamond\alpha^{-1}(j)}$

\paragraph{\label{sec:head-root}Setting the root}
Finally, at each step we make the \emph{new} dummy $d$ the root of the tree by orienting edges away from it. This means the permanents computed will be found in
$P[d](\rho'(d), \varnothing)$.
Together with the correct order of permanents, this allows us to \emph{reuse previously computed $P$-tables in subsequent steps}. We show the idea in \Cref{fig:tree-rotations}, where we highlight the nodes for which we need to compute new $P$-tables. Recall that $Q$-tables stay the same (except for the new dummy node).

Using the example, observe that it is only in the first step ($j = 1$, \cref{fig:tree-rotations:original}) where we need to compute all of the $P$-tables. Going to step $j=2$ (\cref{fig:tree-rotations:R1}), neither the node $t_3$ nor its subtree $T_{t_3}$ changed, and their $P$-tables stay the same. We must (re)compute only the tables $P[t_1]$ whose subtree $T_{t_1}$ changed, and $Q[d]$ along with $P[d]$ corresponding to a new node.
Observe that when going to step $j=3$ (\cref{fig:tree-rotations:R2}), the contents and subtree of $t_1$ stay the same. This generalizes to the left subtree $t_1, \dots, t_{j-2}$ staying the same for every step $j>2$, requiring no new table computation. We thus conclude that in every step $j > 1$, we need \emph{only compute at most three new tables}.\footnote{Only one table if, in each step $j$, instead of adding a dummy $d$, we only remove $t_j$ and root the tree at $t_{j-1}$.}

\subsubsection{Post-processing}

Finally, having generated an outcome $\vec r = (r_1, \dots, r_n)$ of size $n$, we follow the Clifford and Clifford~\cite{Cliffords2018} algorithms: forget the identity of photons, and obtain the sample $\vec n'$ (vector of mode occupations) by counting how many times each outcome mode appears in $\vec r$.

\subsection{\label{sec:runn-time-asympt} Running time asymptotic analysis}

To compute the running time, we note first that the computation of a $Q$-table takes time $\cO(2^\omega \omega^2)$, and that of a $P$-table of a node with $c$ children takes time $\cO(c\, 2^\omega \omega^2)$~\cite{Cifuentes2015}. We derive these results in \Cref{app:CP-running-time}, specifically Lemma~\ref{lem:app:runtime-Q} and Lemma~\ref{lem:app:runtime-P}, respectively. Both of these bounds depend on the size of the contents of the corresponding tree node, which is bounded by the treewidth $\omega$ of $\mathrm{T^C}$. Even though we work ubiquitously with restrictions which may decrease treewidth, we still keep this general bound on the treewidth for simplicity.

In the following, we focus on the sampling of marginals $k=2,\dots,n$ which subsumes the preparation, sampling the first marginal ($k=1$), as well as the postprocessing.

To prepare for the Laplace expansion~(\ref{new:eq:laplace-expansion-per}) with a given output $r_k$, we iterate through the tree as described in the \Cref{sec:moving-head}, computing $k$ permanents.
For the first step $j=1$, we compute the $P$-tables for all $k$ nodes which have each at most one child, taking time $\cO(k 2^\omega \omega^2)$.
Then at each step $j>1$, we only compute up to three new dynamic programming tables $P[t_{j-1}]$, $Q[d]$, and $P[d]$, where $t_{j-1}$ has at most one child, and $d$ has at most two. Step $j>1$ thus takes time $\cO(2^\omega \omega^2)$.
In total, the computation of the $k$ permanents needed in Laplace expansion takes the time of $\cO(k2^\omega \omega^2)$.

Sampling $r_{k}$ takes time $\cO(k^2 \omega)$ as the Laplace expansion needs $k$ scalar multiplications per possible output, and there are at most $k \omega$ non-zero outputs. Here, we recall that each column has at most $\omega$ rows with nonzero elements, and at this $k$-th step, we have access to $k$ of those. Note also that this is an upper bound: some columns may be overlapping, in which case the number of possible values of $r_{k}$ will be lower.

The loop that adds a new photon runs for $k$ in range from $2$ to $n$. Putting this all together, generating a single sample using our algorithm needs the following asymptotic running time:
\begin{equation*}
  \sum_{k=2}^n \cO(k2^\omega \omega^2) + \cO(k^2 \omega)
  = \cO(n^22^\omega \omega^2) + \cO(\omega n^3).
\end{equation*}
Remark that $\cO(\omega n^3)$ remains better than $\cO(m n^2)$ while $\omega n < m$.
If $\omega n > m$ it is easy to see that by design of our algorithm we recover
the scaling $\cO(m n^2)$.

For completeness, we shortly remark that the tree manipulations, i.e. column permutations and restrictions of tree decompositions, which consist of computing set intersections, are subsumed by the above complexity.

\section{\label{sec:concl-future-work}Conclusion and future work}

We have presented an algorithm that generates samples of $n$ photons from a Boson Sampling experiment on a shallow circuit in the no-collision regime $m = \Omega(n^2)$, a necessary constraint in our algorithm. Combining the central idea of the Boson Sampling simulation work of Clifford and Clifford and the computation of permanent of matrices of with fixed treewitdth of Cifuentes and Parrilo, we have constructed an algorithm that generates a sample of $n$ photons from a Boson Sampling experiment in an $m$-mode circuit described by a matrix of treewidth $\omega$ that runs in time $\cO(n^22^\omega \omega^2) + \cO(\omega n^3)$.
For the case of shallow circuit, with depth $D = c \log n$,  the sampling algorithm becomes polynomial time $\cO(n^{2(c+1)}\log^2 n) + \cO(n^3 \log n)$, as suggested in \cite{Garcia-Patron2019}.

A straightforward use of the Cifuentes and Parrilo result leads to a scaling of $\cO(m n^2 2^\omega
\omega^2)$ or $\cO(n^3 2^\omega \omega^2)$. The former result was attained in \cite{Oh2021} by replacing the permanent computations by the algorithm CP, but without optimizing the Laplace expansion on the tree decomposition.
Our key technical improvement is to restore the idea of Clifford and Clifford to decouple the dominant running time term from the number of output modes $m$
using the Laplace expansion and precomputing submatrices~$\W_{\diamond j}$
in the same running time as needed to compute a permanent of size~$k$. We adapt this idea to the tree decomposition framework of Cifuentes and Parrilo by performing a series of local updates to single nodes that allows us to compute $k$ permanents of subtrees in time comparable to the computation of the permanent of the full tree.

In the regime $\omega=\cO(n)$ our sampling algorithm has an extra polynomial cost with respect
to the traditional Clifford and Clifford algorithm. An interesting open question would be to find an algorithm that has better scaling than ours and that converges to the exact running time
$\cO(n2^n)$ when $\omega=n$. Another interesting open problem would be to generalize our result to the collision scenario, where multiple photons can be measured in the same mode. On face value, the latter complicates the use of tree decompositions, because such collisions cause the copying of rows of the matrix whose permanent represents the amplitude, creating a new cycle in its graph.
Thus we expect that one would likely need to construct entirely new tree decompositions in such cases.

Finally, it would be interesting to see whether our ideas can be adapted to circuits that allow nonlocal operations, such as the scheme of~\cite{go_ExploringShallowDepthBoson_2024}. These can obtain a dense unitary with only logarithmic circuit depth, posing a seemingly difficult challenge for tree-decomposition based methods that rely on matrix sparsity.
Perhaps if one were able to reintroduce sparsity, for example by exploiting loss to truncate small matrix entries, one may be able to find a good tree decomposition. In such case, it should be possible to either adapt our method, or perhaps use our work as a toolkit to develop a different approach to tackle this problem.

\begin{acknowledgments}
  This work is partially based on work done by S.N. as his 2022 undergraduate final project~\cite{dissertation} supervised by \mbox{R.G.-P.}, and continued as an LFCS internship over the summer 2022. R.G.-P. was supported by the EPSRC-funded project Benchmarking Quantum Advantage. S.N. thanks William Clements for helpful discussions.
\end{acknowledgments}

\bibliography{bibliography}
\appendix
\section{\label{sec:boson-sampling}Boson Sampling}

Here, we expand the details on quantum optics and Boson Sampling omitted in \Cref{sec:bs-cc}.
Recall that in an ideal Boson Sampling circuit, the input and output can be written in a basis of quantum states corresponding to the location of $n$ photons in $m$ modes, represented as \emph{number (Fock) states} $\ket{\vec n} = \ket{n_1, \dots, n_m} \in (\C^n)^{\otimes m}$, where $n_i \in \N$ is the number of photons in mode $i$, and $\sum_i n_i = n$. These states are canonically written using the creation operators as
\begin{equation*}
  \ket{\vec n} = \prod_{i=1}^m \frac{(\hat a_i^\dagger)^{n_i}}{\sqrt{n_i!}} \; \ket{0}^{\otimes m},
\end{equation*}
where $\hat a_i^\dagger$ acts only on mode $i$. An \emph{interferometer} is a linear optical device that is represented by a unitary matrix $\U \in \mathbb{U}(m)$ acting linearly on the $m$ creation operators as
\begin{equation}
  \label{eq:interferometer-action-on-operators}
  \hat a_i^\dagger \xmapsto\U \sum_{j=1}^m \U_{i,j} \hat a_j^\dagger
  \qquad
  \forall i = 1, \dots, m.
\end{equation}
The unitarity of $\U$ (i.e. $\U^\dagger \U = \U \U^\dagger = \I$) ensures the total number of photons in preserved \cite{kok2010}.

\subsection{\label{sec:1stquant}First quantization representation}

We now reformulate Boson Sampling in first quantization, used in \Cref{sec:intro-bs}.
For a given occupation state $\ket{\vec n}$, we start by writing an~$n$-qudit state where each $m$-dimensional qudit represents a photon in a specified mode, and these are arranged in a non-decreasing order:
\begin{equation}
  \label{eq:nondecreasing-qudit-product}
  \ket{\vec z}
  \deq
  \ket1^{\otimes n_1} \otimes \ket2^{\otimes n_2} \otimes \cdots \otimes \ket{m}^{\otimes n_m}.
\end{equation}
Any permutation $\sigma$ of the $n$ qudits will provide a new basis state $\ket{\vec r}=\sigma\ket{\vec z}$ that has the same mode occupation pattern as $\ket{\vec z}$, given by $\ket{\vec n}$. The states corresponding to the same number state are considered equivalent, written $\ket{\vec z} \sim \ket{\vec r}$.
The size of each equivalence class of states $\ket{\vec r}$ with the same photon distribution as $\ket{\vec z}$ is given by
\begin{equation}
  \label{eq:class-size-of-z-r}
  \binom{n}{n_1, \dots, n_m} = \frac{n!}{\prod_{i=1}^m n_i!}.
\end{equation}
In the notation of $\ket{\vec z}$ and $\ket{\vec r}$, the qudits are assigned specific locations and are therefore distinguishable. To ensure indistinguishability, the state $\ket{\Phi_{\vec n}}$ corresponds to the symmetric superposition of all their permutations:
\begin{equation}
  \label{eq:qudit-representation}
   \ket{\Phi_{\vec n}} =
  \frac{1}{\sqrt{n! \prod_{i=1}^m n_i!}} \sum_{\sigma \in \Sg_n} \sigma \ket{\vec z},
\end{equation}
where $\sigma$ is a permutation in the symmetric group on $n$ elements  $\Sg_n$, and the prefactor on the RHS is its normalization~\cite{Moylett2018}. Note this is eq.~\eqref{new:eq:qudit-representation}.

The action of the interferometer $\U$ in the first-quantization representation follows from (\ref{eq:interferometer-action-on-operators}). For a single qudit $\ket i$, the action reads:
\begin{equation}
  \ket{i} \xmapsto\U \sum_{j=1}^m \U_{i,j} \ket{j} = \U \ket{i},
\end{equation}
which results from its action on $\hat a^\dagger_i$ and the fact that
$|i\rangle=\hat{a}^\dagger_i|0\rangle$.
Likewise, product states correspond to products of creation operators,which give the relation
\begin{equation}
  \label{eq:U-action-on-z}
  \ket{\vec z} \equiv \bigotimes_{k=1}^n \ket{z_k}
  \xmapsto{\U}
  \bigotimes_{k=1}^n (\U \ket{z_k}) = \U^{\otimes n} \ket{\vec z},
\end{equation}
where we have introduced the notation $\ket{z_j}$ to mean the $j$-th qudit in the product \eqref{eq:nondecreasing-qudit-product}.
By linearity we obtain that $\ket{\Phi_{\vec n}} \mapsto \U^{\otimes n} \ket{\Phi_{\vec n}}$,
i.e. in the first quantization representation, the linear optical circuit unitary over $n$
photons acts as $\U^{\otimes n}$.

\subsection{\label{sec:boson-sampl-outcome-probabilities}Boson Sampling outcome probabilities}

We fill in the derivation behind equation (\ref{eq:outcome-probability-permanent-expanded}), which follows the steps in \cite{Moylett2018}.
We represent $\ket{\vec n} \cong \neket{\Phi_{\vec n}}$ canonically by $\ket{\vec z}$ defined in (\ref{eq:nondecreasing-qudit-product}), we analogously represent the output state $\ket{\vec n'}$ in qudit space by $\ket{\vec z'}$, again with tensor factors ordered in non-decreasing order. We denoted $\ket{z_j}$ the $j^{\mathrm{th}}$ photon (tensor factor) of $\ket{\vec z}$, and we define $\ket{z'_j}$ for $\ket{\vec z'}$ analogously.

Recall from (\ref{eq:U-action-on-z}) that in the qudit representation, the interferometer $\U$ acts as the $n$-fold tensor product $\U^{\otimes n}$. The probability of measuring the output state $\ket{\vec{n}'}$ is given by the Born rule $\Pr\left[ \vec{n}' \middle| \vec{n} \right]  = \left| \braoket{\vec{n}'}{\U}{\vec{n}} \right|^2$. We write this in the qudit representation:
\begin{widetext}
  \begin{align*}
    \Pr \left[ \vec{n}' \middle| \vec{n} \right]
    & \overset{\hphantom{(A4)}}{=} \left| \braoket{\Phi_{\vec n'}}{\: \U^{\otimes n}}{\Phi_{\vec n} \:} \right|^2 \\
    & \overset{(\ref{eq:qudit-representation})}{=} \left|
      \left(
      \frac{1}{\sqrt{n! \prod_{i=1}^m n'_i!}}
      \sum_{\sigma \in \Sg_n} \bra{\vec z'} \sigma^{-1}
      \right)
      \U^{\otimes n}
      \left(
      \frac{1}{\sqrt{n! \prod_{i=1}^m n_i!}}
      \sum_{\tau \in \Sg_n} \tau \ket{\vec z}
      \right) \right|^2 \\
    &
      \overset{\hphantom{(A4)}}{=}
      \left|
      \frac{1}{n! \sqrt{\prod_{i=1}^m n_i! \; n'_i!}}
      \sum_{\sigma, \tau \in \Sg_n}
      \left(
      \bigotimes_{j=1}^n \bra{z'_j}
      \right) \sigma^{-1}
      \U^{\otimes n}
      \tau
      \left(
      \bigotimes_{k=1}^n \ket{z_k}
      \right)
      \right|^2,
  \end{align*}
\end{widetext}
where we expanded the definition of $\ket{\vec n}, \ket{\vec n'}$, rearranged the sums over the permutations~$\sigma$ and~$\tau$, and decomposed $\ket{\vec z}, \ket{\vec z'}$ into their individual qudits. Note that the action of the permutations is by braiding, rearranging tensor factors; this means that the adjoint corresponds to the inverse: \mbox{$\sigma^{\dagger} = \sigma^{-1}$}, as in $\nebra{\Phi_{\vec n'}}$ above.

The permutation $\sigma^{-1}$, acting as a braiding of the tensor factors, commutes with the operator $\U^{\otimes}$, so we can write:
\begin{widetext}
  \begin{align*}
    \Pr \left[ \vec{n}' \middle| \vec{n} \right]
    &
      =
      \left|
      \frac{1}{n! \sqrt{\prod_{i=1}^m n_i! \; n'_i!}}
      \sum_{\sigma, \tau \in \Sg_n}
      \left(
      \bigotimes_{j=1}^n \bra{z'_j}
      \right)
      \U^{\otimes n}
      \underbrace{\sigma^{-1} \tau}_{\zeta}
      \left(
      \bigotimes_{k=1}^n \ket{z_k}
      \right)
      \right|^2 \\
    &
      =
      \left|
      \frac{n!}{n! \sqrt{\prod_{i=1}^m n_i! \; n'_i!}}
      \sum_{\zeta \in \Sg_n}
      \left(
      \bigotimes_{j=1}^n \bra{z'_j}
      \right)
      \U^{\otimes n}
      \left(
      \bigotimes_{k=1}^n \ket{z_{\zeta^{-1}(k)}}
      \right)
      \right|^2
  \end{align*}
\end{widetext}
were we define the composite $\zeta = \sigma^{-1} \tau$ in $\Sg_n$. The composite ranges over the whole group $\Sg_n$, so we remove one summation and get a factor of $n!$ coming from the standard group-theoretic fact that every $\zeta$ is a composite of $n!$ different pairs of $\sigma$ and $\tau$.
Photons in the product $\bigotimes_k \ket{z_k}$ are reordered by $\zeta$: the $k$-th photon is sent to $\zeta(k)$. We rewrite this state with the permutation inside. To find which is the $k$-th factor after the permutation, we apply the inverse $\zeta^{-1}$ inside the state to obtain $\bigotimes_k \ket{z_{\zeta^{-1}(k)}}$.

Finally, we match up the tensor factors across the expression, giving us terms with $j = \zeta^{-1}(k)$. This reduces to a product of bra-kets with $\U$ in the middle. To bring the expression into a form of a permanent, we replace the permutation $\zeta$ by its inverse $\pi = \zeta^{-1}$. This reorders the summands, but the summation commutes, so the value is the same. We obtain:
\begin{align*}
  \Pr \left[ \vec{n}' \middle| \vec{n} \right]
  & =
    \left|
    \frac{1}{\sqrt{\prod_{i=1}^m n_i! \; n'_i!}}
    \sum_{\pi \in \Sg_n} \;
    \prod_{j=1}^{n} \;
    \V^{\vec n', \vec n}_{j, \pi(j)}
    \right|^2 \\
  & =
    \frac{\big| \per \V^{\vec n', \vec n} \big|^2}{\prod_{i=1}^m n_i! \, n'_i!},
\end{align*}
These are equations~(\ref{eq:outcome-probability-permanent-expanded},~\ref{eq:outcome-probability-permanent-general}) in \Cref{sec:outc-prob}. Recall that we define the matrix $\V^{\vec n', \vec n}$ with components
\[
  \V^{\vec n', \vec n}_{l,h} \deq \braoket{z'_l}{\:\U\:}{z_h},
\]
i.e. the matrix elements of $\U$ at row (output) $z'_l$ and column (input) $z_h$.

\section{\label{sec:clifford}Clifford and Clifford algorithms}

In \Cref{new:sec:algorithms-cca-ccb}, we shortly sketched the main ideas of algorithms CC-A and CC-B. We now elaborate, showing the conceptual steps required to construct these algorithms.

\subsection{\label{sec:vari-orig-phot}Variable origins of photons and algorithm CC-A}

The main idea of CC-A is to perform the sampling in the $\neket{\vec r}$ basis and grow the sample photon-by-photon, using a standard chain rule procedure.
The chain rule states that:
\begin{equation}
    p(r_1,\dots,r_n) = p(r_1)p(r_2|r_1) \cdots p(r_n|r_{n-1}, \dots, r_1).
\end{equation}
If one has access to all marginals of $p(r_1,r_2,...,r_n)$, one can generate a sample using the chain rule by first computing $p(r_1)$ for all $r_1 \in [m] = \{ 1,\dots,m \}$, and committing to a choice of $r_1$ following the computed probability distribution. Having fixed $r_1$, we compute $p(r_1, r_2)$ for all $r_2 \in [m]$, using $p(r_2|r_1)=p(r_1,r_2)/p(r_1)$, and again commit to a choice of $r_2$ following that distribution. Here, we use the value $p(r_1)$ computed in the previous step. Iterating this procedure, we can generate one sample via the computation of $mn$ marginals.

\subsubsection{\label{sec:marg-prob-algor-A}Marginal probabilities}

We follow the steps in~\cite{Cliffords2018} to derive the marginal probabilities.
Equation (\ref{eq:outcome-probability-permanent-general}) gives the probability of outcome $\ket{\vec n'}$ for an arbitrary input $\ket{\vec n}$. We follow the convention that input modes are occupied by up to one photon ($n_i \le 1$), and without loss of generality use the standard input $\ket{\vec n} = \ket{1}^{\otimes n} \otimes \ket{0}^{\otimes (m - n)}$, i.e. the photons are found in the first $n$ modes. This allows us to simplify notation in the following, omitting the reference to $\ket{\vec n}$:
\begin{equation*}
  \Pr \left[ \vec{n}' \right]
  =
  \frac{\left| \vphantom{\big|} \per \V^{\vec n', (1^n, 0^{m-n})} \right|^2}{\prod_{i=1}^m n'_i!}
\end{equation*}

Define $\widetilde\V^{(r_1, \dots, r_k), C}$ which fulfils the same function as $\V^{\vec n', \vec n}$ but it is given by listing which modes appear in the qudit representation, and allows partial samples. Its components are:
\[
  \widetilde\V^{(r_1, \dots, r_k), C}_{i,j} = \braoket{r_i}{\U}{c_j},
\]
where the input qudits are given as an (unordered) set $C = \{ c_1, \dots, c_k \}$, so we impose $c_1 \le c_2 \le \cdots \le c_k$ for convenience.
We can choose any ordering here, because the permanent is invariant under permutations of columns, as well as rows.
Note that we always give the output photons $(r_1, \dots, r_k)$ in the order in which they are sampled. The latter ensures that the newest photon, for which we will comput the marginal probability later, corresponds to the last row of the matrix.
Note further that only the presence or absence of a row or column is relevant (or, more generally, its number of repetitions), but we choose an ordering so that we can write down concrete matrices. The set $C$ is the set of input modes, i.e. columns of $\U$, and the (full) output state is represented by $\ket{\vec z'} = \ket{z'_1} \otimes \cdots \otimes \ket{z'_n}$. We write the probability in this notation, and define the pmf $q$:
\begin{equation}
  \label{eq:probability-psiprime-preparation}
  q(\vec z') \deq
  \frac{\left| \per \widetilde\V^{(z'_1, \dots, z'_n), [n]} \right|^2}{\prod_{i=1}^m n'_i!}
  = \Pr \left[ \vec n' \right].
\end{equation}
where we recall the notation $[n] = \{ 1, \dots, n \}$.

By rewriting the probability in terms of qudits, we can define a new pmf where the photon samples are given in the $\ket{\vec r}$ basis. That is, we now work with \emph{distinguishable} photons, specifying which photon from the input ends up in which output mode. We write the new pmf $p$ defined on all $\vec r \in [m]^n$:
\begin{equation}
  \label{eq:def-p}
  p(\vec r) \deq \Pr[\vec r] =
  \frac{\bigl| \per \widetilde\V^{\vec r, [n]} \bigr|^2}{n!}.
\end{equation}
The permanent is invariant under permutation of rows and columns, so $p(\vec r) = p(\vec z')$, where $\vec z'$ are the same modes in non-decreasing order, representing the same number state, i.e. $\ket{\vec r} \sim \ket{\vec z'}$. We recall from (\ref{eq:class-size-of-z-r}) that there are $\frac{n!}{\prod_{i=1}^m n'_i!}$ such vectors $\vec r$. The two pmfs are related as:
\begin{align*}
  \sum_{\vec r \sim \vec z} p(\vec r)
  & \overset{\hphantom{(A3)}}{=} p(\vec z') \sum_{\vec r \sim \vec z'} 1 \\
  & \overset{(\ref{eq:class-size-of-z-r})}{=} \frac{n!}{\prod_{i=1}^m n'_i!} \cdot p(\vec z') \\
  & \overset{(\ref{eq:def-p})}{\underset{}{=}} \frac{\cancel{n!}}{\prod_{i=1}^m n'_i!} \cdot
    \frac{\bigl| \per \widetilde\V^{\vec z', [n]} \bigr|^2}{\cancel{n!}} \\
  & \overset{(\ref{eq:probability-psiprime-preparation})}{=} q(\vec z'),
\end{align*}
which is the result we expect, and we conclude:
\begin{equation}
  \label{eq:p-relation-to-q}
  p(\vec r) = p(\vec z') = \frac{\prod_{i=1}^m n'_i!}{n!} \cdot q(\vec z').
\end{equation}
This allows us to compute $p$ instead of $q$.

Using the qudit representation and working in the $\ket{\vec r}$ basis, accessible thanks to the pmf $p$, allows us to work with partial samples of $k < n$ photons. However, the definition of $\widetilde\V^{(z'_1, \dots, z'_k), C}$ requires us to select which of the incoming photons were measured in the partial sample. The probability of measuring some partial sample $\vec r^{(k)} \deq (r_1,\dots,r_k) \in [m]^k$ of $k$ photons has to sum over the $\binom{n}{k} = \frac{n!}{k! (n-k)!}$ possible subsets of input modes from which these photons came. We arrive at the following:
\begin{lemma}[partial sample]
  \label{lemma:partial-sample-qudits}
  The marginal pmf of a partial sample of $k \in [n]$ photons is
  \begin{equation}
    \label{eq:marginal-probability-rk}
    p\bigl( \vec r^{(k)} \bigr)
    = \frac{(n-k)!}{n!}
    \sum_{\substack{C \subseteq [n] \\ |C| = k}}
    \left|
      \per \widetilde\V^{\vec r^{(k)}, C}
    \right|^2
  \end{equation}
\end{lemma}
\begin{proof}
  To get the marginal pmf, we sum the full pmf over the part of sample $\vec r$ that is not in $\vec r^{(k)}$, which we call $\vec r' = (r_{k+1},\dots,r_{n})$:
  \begin{align*}
    p\bigl( \vec r^{(k)} \bigr)
    & = \sum_{\vec r'} p(\vec r)
      \overset{(\ref{eq:def-p})}{=} \sum_{\vec r'} \frac{\bigl| \per \widetilde\V^{\vec r, [n]} \bigr|^2}{n!} \\
    &
      = \frac{1}{n!} \sum_{\vec r'}
      \left(
      \sum_{\sigma \in \Sg_n}
      \prod_{j=1}^n \U_{r_j, \sigma_j}
      \right)
      \left(
      \sum_{\tau \in \Sg_n}
      \prod_{\ell=1}^n \U_{r_\ell, \tau_\ell}
      \right)^*
  \end{align*}
  where we expanded the definition of permanent, and  expressed the magniture squared \mbox{$|x|^2 = x x^*$} for any $x \in \C$. Recall that $\widetilde\V^{\vec r, [n]}$ is the matrix with rows $r_1, \dots, r_n$ and columns $[n]$ of $\U$. Below, we write $\sigma_j \deq \sigma(j)$, likewise for $\tau$, to shorten the notation.
  \begin{align}
    \nonumber
    p(\vec r^{(k)})
    &
      = \frac{1}{n!}
      \sum_{\vec r'}
      \sum_{\sigma, \tau \in \Sg_n}
      \prod_{j=1}^n
      \U_{r_j,\sigma_j} \U^*_{r_j,\tau_j} \\
    \label{eq:marginals-prob-matrix-products}
    & = \frac{1}{n!}
      \sum_{\vec r'}
      \sum_{\sigma, \tau \in \Sg_n}
      \prod_{j=1}^n
      \U^\dagger_{\tau_j,r_j} \U_{r_j,\sigma_j}
  \end{align}
  We have rearranged the expressions. In (\ref{eq:marginals-prob-matrix-products}), we used the fact that $\U_{r_j,\tau_j}^* = \U^\dagger_{\tau_j,r_j}$. We move the terms not summed over outside of the summation over $\vec r' = (r_{k+1},\dots,r_n)$:
  \begin{widetext}
    \begin{align}
      \nonumber
      p(\vec r^{(k)})
      & \nonumber
        =
        \frac{1}{n!}
        \sum_{\sigma, \tau \in \Sg_n}
        \left(
        \prod_{j=1}^k \U^\dagger_{\tau_j,r_j} \U_{r_j,\sigma_j}
        \right)
        \sum_{\vec r'}
        \prod_{\ell=k+1}^n
        \U^\dagger_{\tau_\ell,r_\ell} \U_{r_\ell,\sigma_\ell} \\
      & \label{eq:exchange-sum-product}
        =
        \frac{1}{n!}
        \sum_{\sigma, \tau \in \Sg_n}
        \left(
        \prod_{j=1}^k \U^\dagger_{\tau_j,r_j} \U_{r_j,\sigma_j}
        \right)
        \prod_{\ell=k+1}^n
        \underbrace{\sum_{r_\ell=1}^m
        \U^\dagger_{\tau_\ell,r_\ell} \U_{r_\ell,\sigma_\ell}}_{[\U^\dagger \U]_{\tau_\ell, \sigma_\ell}}
    \end{align}
  \end{widetext}
  In (\ref{eq:exchange-sum-product}), we exchange the product and the sum by distributivity. Only one variable ($r_\ell$) appears now, but the product ranges over $\ell$, so we still have all the terms. Note that the summation over $r_\ell$ is a matrix multiplication $\U^\dagger \U$. As $\U$ is unitary, this product is the identity. The product over $\ell$ then becomes $\prod_{\ell=k+1}^n \delta_{\tau_\ell, \sigma_\ell}$, where $\delta_{i,j}$ is the Kronecker delta. The only nonzero terms are those where $\sigma_\ell = \tau_\ell$ for all $\ell>k$. We partition the permutations as $\sigma = \mu \oplus \zeta$ and $\tau = \nu \oplus \zeta$ into a shared part $\zeta \in \Sg_{n-k}$ (this is the part that ensures $\sigma_\ell = \tau_\ell$ for $\ell > k$), and $\mu,\nu \in \Sg_k$, which may be different. This means we also partition the set $[n]$ on which $\sigma,\tau$ act into a subset $C$ of size $k$ on which $\mu,\nu$ act, and $[n]\setminus C$ acted upon by~$\zeta$.%
  \footnote{By a partition $\sigma = \mu \oplus \zeta$, we mean that $\sigma_\ell = \mu_\ell$ for $i \le k$ and $\sigma_\ell = \zeta_{\ell-k} + k$ for $\ell > k$.
    Note we abuse notation by saying $\sigma = \mu \oplus \zeta$ and $\tau = \nu \oplus \zeta$, hiding the detail that the output may need to be rearranged by some shared shared $\lambda \in \Sg_n$ as it may not fit into the blocks $[k]$ and $[n-k]$. In reality, we have the composites $\sigma = \lambda (\mu \oplus \zeta)$ and $\tau = \lambda (\nu \oplus \zeta)$. However, these are formal details that would only obscure the meaning of the steps.}

  We rewrite the summation over $\sigma$ and $\tau$ as summations over the selection of $C \subseteq [n]$, i.e. over the partitions, such that we select the part of $[n]$ where $\mu,\nu$ act, and as summations over $\mu,\nu$ themselves. For each $C$, there are $(n-k)!$ possible permutations $\zeta$; this gives a new factor in~(\ref{eq:marginal-almost-there}):
  \begin{widetext}
    \begin{align}
      \label{eq:marginal-almost-there}
      p\bigl( \vec r^{(k)} \bigr)
      & =
      \frac{(n-k)!}{n!}
      \sum_{\substack{C \subseteq [n] \\ |C| = k}} \;
      \underbrace{\left(
        \sum_{\mu \in \Sg_C}
        \prod_{j=1}^k \U_{r_j,\mu_j}
      \right)}_{\per \widetilde\V^{(r_1, \dots, r_k), C}}
      \underbrace{\left(
        \sum_{\nu \in \Sg_C}
        \prod_{\ell=1}^k \U_{r_\ell,\nu_\ell}
        \right)^*}_{\left( \per \widetilde\V^{(r_1, \dots, r_k), C} \right)^*}\\
      \nonumber
      & = \frac{(n-k)!}{n!}
        \sum_{\substack{C \subseteq [n] \\ |C| = k}}
        \left|
        \per \widetilde\V^{\vec r^{(k)}, C}
        \right|^2
    \end{align}
  \end{widetext}
  Thus we have recovered the permanent of submatrices, and this concludes the proof.
\end{proof}

\subsubsection{Running time}

The algorithm CC-A samples from $p(r_1, \dots, r_n)$ by committing to the partial sample $(r_1, \dots, r_{k-1})$ and extending it by a single $r_{k}$ at each step $k=2,\dots,n$. For every $k$ we need to compute $m$ marginals, each consisting of a sum of size $\binom{n}{k}$ and a permanent of a matrix of size $k \times k$ taking the time $\cO(k2^k)$ steps each using the algorithm by Ryser~\cite{MR0150048} or Glynn~\cite{Glynn2010}. This leads to a total running time:
\begin{equation}
  \cO\left( m\sum_{k=1}^n k2^k\binom{n}{k} \right) = \cO( mn3^n ).
\end{equation}
At each step $k$, we need to remember the values $\{ p(r_k | r_1, \dots, r_{k-1}) \}_{r_k=1,\dots,m}$, i.e.~store a vector of $m$ values.
Therefore we require $\cO(m)$ additional memory, less than the memory needed to store the matrix $\U$ of size $\cO(m^2)$.

\subsection{\label{sec:input-perm-trick} Removing the costly sum: algorithm CC-B}

The first improvement of CC-A suggested by Clifford and Clifford~\cite{Cliffords2018}, which we refer to as CC-B, is to
remove the costly sum over $\binom{n}{k}$ terms corresponding to sets $C$ by preselecting them based on a random choice, and only using the inputs (columns) in the chosen $C$.
The idea is to remove the selection from the marginal itself, and instead determine the order in which input photons appear. This order is chosen randomly for each full sample. This means each sample comes from a simpler probability distribution, but the ensemble of many samples follows the correct distribution.

As part of the initialization of the algorithm, i.e. before the sample is generated, choose a uniformly random permutation $\alpha \in \Sg_n$. At each $k$, the set $C$ will be defined to be $C = \alpha([k]) = \{ \alpha(1), \dots, \alpha(k) \}$.
For each $C$ of size $k$, there are $k! (n-k)!$ permutations $\alpha$ such that $C = \alpha([k])$, so we need to divide by this in~(\ref{eq:marginal-probability-rk}) to preserve equality:
\begin{align}
  \nonumber
  p(\vec r^{(k)})
  & =
    \frac{1}{n! \, k!}
    \sum_{\alpha \in \Sg_n}
    \left|
    \per \widetilde\V^{\vec r^{(k)}, \alpha([k])}
    \right|^2 \\
  \label{eq:p-as-E-phi}
  & = \underset{\alpha \in \Sg_n}\Expectation
    \underbrace{\left[
    \frac{1}{k!}
    \left|
    \per \widetilde\V^{\vec r^{(k)}, \alpha([k])}
    \right|^2
    \right]}_{\phi(\vec r^{(k)} | \alpha)}.
\end{align}
We see that the original pmf $p$ is the expectation of a different pmf called $\phi$, taken over the uniformly random $\alpha \in \Sg_n$, where $\phi$ is conditioned on $\alpha$.
The expectation value over many samples, each with a new uniformly random $\alpha$, leads to the correct marginal probability.

For each photon $k$, the $1/k!$ in (\ref{eq:p-as-E-phi}) is a shared constant prefactor for each value of $r_k$. Similarly to CC~\cite{Cliffords2018}, for convenience, we can omit this shared prefactor in the algorithm and sample from unnormalized pmfs. This is possible using the Walker's alias method in time $\cO(m)$ for a pmf stored in an array of length $m$, the latter being the number of possible new modes~\cite{Walker1974NewFM}.\footnote{Precisely, the time per sample is $\cO(1)$ after pre-processing time of $\cO(m)$. In our use case, we need a single sample from each pmf, so we characterize the time as just the total $\cO(m)$.}

\subsubsection{Running time}

The running time of Algorithm CC-B scales as
\begin{equation}
  \cO\left( m\sum_{k=1}^n k2^k \right) = \cO\left( m n\sum_{k=1}^n 2^k \right)
  = \cO(mn2^{n+1}).
\end{equation}
Hence CC-B generates a sample in time $\cO(mn2^n)$. Similarly to CC-A, it requires additional space $\cO(m)$, where $m$ is the number of modes.

\subsection{\label{app:laplace-expansion}Laplace expansion: algorithm CC-C}

In \Cref{sec:lapl-expans-algor}, we describe the Laplace expansion giving the algorithm CC-C, the final improvement of~\cite{Cliffords2018}. For convenience, we give in \Cref{fig:CC-C-float} a pseudocode that summarizes the steps.

\begin{figure}
\begin{algorithm}[H]
  \SetAlgoRefName{CC-C}
  \caption{Boson Sampling}
  \label{alg:cliffords-c}
  \KwIn{
    \begin{itemize}
    \item number of modes $m$, and of photons $n$, $n \le m$
    \item matrix $\U \in \mathbb{U}(m)$ representing the interferometer
    \end{itemize}
  }
  \KwResult{a single sample $\vec z'$}

  $\alpha \gets$ uniformly random permutation from $\Sg_n$ \;

  $\V^{\vec n} \gets$ first $n$ columns of $\U$ permuted using $\alpha$ \;

  initialize $\vec r = (r_1, \dots, r_n) \gets (0, \dots, 0)$ \;

  $w_1(i) \gets |\V^{\vec n}_{i, 1}|^2$ for all $i \in [m]$   \tcp*{first photon}

  $r_1 \gets$ sample from the pmf $w_1$ \;

  \For(\tcp*[f]{rest of the sample}){$k \gets 1, \dots, n-1$}{

    $\W \gets$ rows $r_1, \dots, r_{k}$ and first $k+1$ columns of $\V^{\vec n}$
    \tcp*{corresponging to $\widetilde \V^{(r_1,\dots,r_{k}), \alpha([k+1])}$}

    compute $\{ \per \W_{\diamond i} \}_{i=1}^{k+1}$

    $w_{k+1}(i) \gets \left| \sum_{j=1}^{k+1} \V^{\vec n}_{i, j} \per \W_{\diamond j} \right|^2$ for all $i \in [m]$

    $r_{k+1} \gets$ sample from the (unnormalised) pmf $w_{k+1}$
  }

  $\vec z' \gets$ sort $\vec r$ in non-decreasing order

  \KwRet{$\vec z'$}, which represents $\ket{\Phi_{\vec n'}} = \frac{1}{({\cdots})} \sum_{\sigma \in \Sg_n} \sigma \ket{\vec z'}$
\end{algorithm}
\caption{Pseudocode of algorithm CC-C. Note the differring convention with respect to the main text where we expand from a sample of $k-1$ to $k$ photons. In the present appendix, for notational simplicity, we expand from $k$ to $k+1$ photons.}
\label{fig:CC-C-float}
\end{figure}

\subsubsection{\label{app:CC-C-runtime}Running time}

The running time of CC-C is:
\begin{subequations}
  \begin{align}
    \nonumber
    \sum_{k=1}^n & \cO(k2^k) + \cO(mk) \\
    \label{eq:bound-k-by-n}
    & = \cO\left( n\sum_{k=1}^n 2^k \right) + \cO\left( m \sum_{k=1}^n k \right) \\
    \label{eq:running-time-of-CC-C}
    & = \cO\left( n2^{n+1} \right) + \cO\left( m n^2 \right).
  \end{align}
\end{subequations}
In (\ref{eq:bound-k-by-n}), we use the fact that $k \le n$ to bound $k 2^k = \cO(n 2^n)$. The summation $\sum_{k=1}^n 2^k = 2 (2^n - 1)$, and together these give the term $\cO(n2^{n+1})$ in (\ref{eq:running-time-of-CC-C}).
Algorithm CC-C generates a sample in time $\cO(n 2^n) + \cO(m n^2)$, requiring space $\cO(m)$. To our knowledge, CC-C is currently the best algorithm to simulate Boson Sampling in the regime with no collisions and no additional constraints, e.g. on the depth of the circuit.

\section{\label{sec:appendix-CP}Algorithm CP}

In this section, we elaborate on the Cifuented and Parrilo~\cite{Cifuentes2015} algorithm, denoted CP, from \Cref{sec:permanent-algorithm}, particularly the details of the dynamic programming tables, and the running time. In the following, we write $\per_M(R,C) = \per \restr M{R,C}$ for the permanent of a submatrix of $M$ with rows $R$ and columns $C$, following the notation of~\cite{Cifuentes2015}.

\subsection{\label{sec:appendix-CP:tables}Dynamic programming tables}

\subsubsection{The local $Q$-tables}
\label{sec:Q-tables}

For every node $t \in T$, we have the table of \emph{local permanents} $Q[t](R,C) = \per_M(R,C)$ for all $R \subseteq \rho(t)$ and $C \subseteq \kappa(t)$, i.e. resulting from combinations of row and column vertices contained in the node $t$. These represent submatrices of $M$ with those rows and columns.

\subsubsection{\label{sec:P-tables}The subtree $P$-tables}

Crucial in the construction are the tables of \emph{subtree permanents} $P[t](R,C)$ where $R \subseteq \rho(t)$ and $C \subseteq \kappa(t)$ again, same as in $Q[t]$. Their values are
\begin{equation}
  \label{eq:P-table}
  P[t](R,C) = \per_M (R \cup \overline\Delta{}^\rho_t, C \cup \overline\Delta{}^\kappa_t),
\end{equation}
where $\overline\Delta{}^\rho_t \deq \rho(T_t) \setminus \rho(t)$, and similarly $\overline\Delta{}^\kappa_t\deq \kappa(T_t) \setminus \kappa(t)$. The table $P[t]$ stores the permanents of submatrices represented by the entire subtree $T_t$. Note that the labels of $P[t]$ contain only rows $R \subseteq \rho(t)$ and columns $C \subseteq \kappa(t)$, but the table elements relate to permanents of submatrices with rows and columns coming from the descendants of $t$, ensured by $\overline\Delta_t^\rho$ and $\overline\Delta_t^\kappa$. For leaves $\ell$, we have $P[\ell] = Q[\ell]$, because they have no descendants. Conversely, the root $r \deq \rt(T)$ contains the final result:
\begin{equation}
  \label{eq:root-permanent}
  P[r](\rho(r), \kappa(r)) = \per_M(\rho(T), \kappa(T)) = \per M.
\end{equation}
In (\ref{eq:root-permanent}), the equality follows from the Axiom~\ref{def:treedec-axiom-vertex-cover}, which says that $\rho(T)$ is the set of all rows of the (whole) matrix $M$, and similarly $\kappa(T)$ is the set of all columns.

\begin{note}
  \label{note:Delta-sets-with-withou-bar}
  We remark on the notational distinction between $\overline\Delta{}^{\rho}_t$ (resp. $\overline\Delta{}^{\kappa}_t$) in this section and $\Delta^{\rho}_{c_j}$ (resp. $\Delta^{\kappa}_{c_j}$) in the pseudocode in \cref{fig:pseudocode-CP} in \Cref{sec:permanent-algorithm}. While the former (with bar) corresponds to the entire subtree of $t$, the latter (without bar) corresponds to individual children of $t$ and is defined in the following section.
\end{note}

\subsubsection{Helper tables}

The table $P[t]$ of an internal node $t$ is computed using the $P$-tables of the children of $t$ in a subset convolution. This requires computing helper tables shown below, however, these do not need to be stored for future computation.
The first table $Q'[t | c_j]$ is associated to permanents of submatrices made of rows and columns in $t$ that are shared with its child $c_j$:
\begin{align}
  \label{eq:CP-tables-Q'}
  Q'[t | c_j](R,C)
  & = (-1)^{|R|} \; Q[t](R,C),
\end{align}
where $R \subseteq \Lambda^\rho_{c_j} \deq \rho(c_j) \cap \rho(t)$ and $C \subseteq \Lambda^\kappa_{c_j}\deq \kappa(c_j) \cap \kappa(t)$.

The table $Q''[t \gets c_j]$ brings permanents of submatrices from the child's table $P[c_j]$ up the tree and makes them available to node $t$. The values transferred are those that correspond to submatrices with rows and columns contained in $c_j$ but not in $t$, together with a selection of rows $R \subseteq \Lambda^\rho_{c_j}$ and columns $C \subseteq \Lambda^\kappa_{c_j}$ shared by the two nodes. These index the table:
\begin{equation}
  \label{eq:CP-tables-Q''}
  Q''[t \gets c_j](R,C)
  = P[c_j](R \cup \Delta^\rho_{c_j}, C \cup \Delta^\kappa_{c_j}),
\end{equation}
where $\Delta^\rho_{c_j} \deq \rho(c_j) \setminus \rho(t)$ and $\Delta^\kappa_{c_j} \deq \kappa(c_j) \setminus \kappa(t)$.

We shortly revisit Note~\ref{note:Delta-sets-with-withou-bar}. Recall the $P$-tables in~\eqref{eq:P-table} reference rows $\overline\Delta{}^{\rho}_{t}$ (resp. columns $\overline\Delta{}^{\kappa}_{t}$) corresponding to the whole subtree $T_t$. In~\eqref{eq:CP-tables-Q''}, the helper table $Q''[t \gets c_j]$, the device used to connect $t$ to its subtree, only references the sets $\Delta^\rho_{c_j}$ (resp.~$\Delta^\kappa_{c_j}$) corresponding to the immediate children of $t$. The reason is that the permanents stored in $P[t]$ are obtained recursively using the tree structure, computed from partial permanents stored in the tables $P[c_j]$ of the children~$c_j$ which are computed before \mbox{$Q''[t \gets c_j]$}. This justifies the definition~\eqref{eq:CP-tables-Q''}.

\subsubsection{Subset convolution}

The final step in computing a value of $P[t](R,C)$ for any internal node $t$ (including the root) is to perform the \emph{subset convolution} over $R$ and $C$ of all the prepared tables: the local table $Q[t]$, and the helpers $Q'[t|c_j]$ and $Q''[t \gets c_j]$ of all children $c_j \in \ch(t)$. The subset convolution, given in eq.~(\ref{eq:CP-subset-convolution}) below, computes the value $P[t](R,C)$ as a summation over collections of subsets $\{ R_i^{(\prime,\prime\prime)} \subseteq R \}_i$ and $\{C_i^{(\prime,\prime\prime)} \subseteq C \}_i$ that partition $R$ and $C$, respectively. The formula is:
\begin{widetext}
  \begin{equation}
    \label{eq:CP-subset-convolution}
    P[t](R,C) = \sum_{\{R_i^{(\prime,\prime\prime)}, C_i^{(\prime,\prime\prime)}\}_i} \bigg( Q[t](R_t, C_t)
    \times \prod_{c_j \in \ch(t)}
    Q'[t | c_j](R'_{c_j}, C'_{c_j})
    \times
    Q''[t \gets c_j](R''_{c_j}, C''_{c_j}) \bigg),
  \end{equation}
\end{widetext}
which is labeled (\ref{new:eq:CP-subset-convolution}) in \Cref{sec:permanent-algorithm}.
Substituting the definitions of the helper tables, the above implements the subset convolution from~\cite[Lemma 10]{Cifuentes2015}.

Below, we give a concrete characterization of the subsets $\{R_i^{(\prime,\prime\prime)}, C_i^{(\prime,\prime\prime)}\}_i$, but first we explain at high level the form of the equation. Each summand in (\ref{eq:CP-subset-convolution}) involves all of the tables $Q[t]$, $Q'[t|c_j]$ and $Q''[t \gets c_j]$ (for all children $c_j$ of $t$), and the subsets $\{R_i^{(\prime,\prime\prime)}, C_i^{(\prime,\prime\prime)}\}_i$ index those tables.
We combine values stored in the tables $Q, Q', Q''$ for particular submatrices that are selected using the subsets $\{R_i^{(\prime,\prime\prime)}, C_i^{(\prime,\prime\prime)}\}_i$. These table values are used to compute terms of the permanents stored in $P[t](R,C)$.

Each partition of rows $\{R_i^{(\prime,\prime\prime)} \subseteq R\}_i$ consists of:
\[ R_t \sqcup R'_{c_1} \sqcup R''_{c_1} \sqcup \cdots \sqcup R'_{c_k} \sqcup R''_{c_k} = R, \]
where $k$ is the number of children of $t$.
Similarly each partition of columns $\{ C_i^{(\prime,\prime\prime)} \subseteq C\}_i$ consists of:
\[   C_t \sqcup C'_{c_1} \sqcup C''_{c_1} \sqcup \cdots \sqcup C'_{c_k} \sqcup C''_{c_k} = C. \]

The constraints on what those individual subsets may be are given by the following:
Let $S$ be one of the above tables (e.g. $Q'[t|c_j]$ for some child $c_j$). If for some $\hat R \subseteq R$ and $\hat C \subseteq C$, the value $S(\hat R, \hat C)$ is not defined, we assume $S(\hat R, \hat C) = 0$ or equivalently, we exclude the corresponding partitions from the summation. This reasoning gives us the following contraints on the subsets $\{ R_i^{(\prime,\prime\prime)}, C^{(\prime,\prime\prime)}_i \}_i$ in the partitions for possibly nonzero terms of (\ref{eq:CP-subset-convolution}):
Recall that $Q[t]$ is defined for sets of rows and columns contained in $t$, so
\begin{align*}
  R_t
  & \subseteq \rho(t)
  & \text{and} &
  & C_t
  & \subseteq \kappa(t).
\end{align*}
Similarly, $Q'[t | c_j]$ and $Q''[t \gets c_j]$ are defined for
\begin{align*}
  R'_{c_j}, R''_{c_j}
  & \subseteq \Lambda^\rho_{c_j}
  & \text{and} &
  & C'_{c_j}, C''_{c_j}
  & \subseteq \Lambda^\kappa_{c_j},
\end{align*}
respectively.

\subsection{\label{app:CP-running-time}Running time details of CP}

We now derive the running time of algorithm CP, following \cite[Lemma 2, Theorem 12]{Cifuentes2015}.
In the following, we have a matrix $M \in \C^{n \times n}$, a tree decomposition $\mathrm{T} = (T, \rho, \kappa)$ of its graph $G(M)$, and this tree decomposition has treewidth $\omega$.
\begin{lemma}
  \label{lem:app:runtime-Q}
  For a node $t \in T$, computing the table $Q[t]$ takes time $\cO(2^\omega \omega^2)$.
\end{lemma}
\begin{proof}
  (See Lemma 2 in \cite{Cifuentes2015}.)
  We compute the table $Q[t]$ for pairs of subsets $R \subseteq \rho(t)$ and $C \subseteq \kappa(t)$, with values $Q[t](R,C) = \per_M(R,C)$ by Laplace expansion:
  \[ \per_M(R,C) = \sum_{c \in C} M_{r,c} \per_M(R \setminus \{r\}, C \setminus \{c\}), \]
  where $r \in R$ is a selected row (we choose $r = \min R$ for convenience). Note that $\per_M(\varnothing, \varnothing) = 1$ and $\per_M(R,C) = 0$ if the sizes do not match, i.e. $|R| \ne |C|$. Thus we only need to compute permanents of square submatrices selected from the node $t$. We do this by starting with the $1 \times 1$ matrices, and then increase size, on each step using the Laplace expansion.

  There are $\cO(2^{|\rho(t)| + |\kappa(t)|}) = \cO(2^\omega)$ values to compute: this is the size of the domain, though it is an upper bound because as mentioned above, some of the entries are always zero. Noting that $\omega$ is the treewidth, any node $t \in T$ has the size of its contents $|\rho(t)| + |\kappa(t)| - 1 \le \omega$. For each of those steps, we need $\cO(|\kappa(t)|) = \cO(\omega)$ operations -- this is the number of columns in the summation. Finally, for each summand, we need $\cO(\omega)$ time to find the sets $R \setminus r$, $C \setminus c$ and look up the stored value of the corresponding permanent.
  Thus, the time is bounded as $\cO(2^\omega \omega^2)$.
\end{proof}

\begin{lemma}
  \label{lem:app:runtime-P}
  For an internal node $t \in T$, computing the table $P[t]$ takes time $\cO(k_t 2^\omega \omega^2)$, where $k_t$ is the number of children of $t$.
\end{lemma}
\begin{proof}
  (See Theorem 12 in \cite{Cifuentes2015}.)
  First note that the helper tables $Q'[t|c_j]$ and $Q''[t \gets c_j]$ for each child $c_j$ of $t$ simply reference $Q[t]$ and $P[c_j]$, in the latter case with an inexpensive set union in the argument. The relevant computation happens in the subset convolution, which is done in time $\cO(k_t 2^\omega \omega^2)$, where Cifuentes \& Parrilo adapt the method from \cite{Bjorklund}.
\end{proof}

\begin{theorem}[running time of CP]
  We can compute $\per M$ in time $\cO(n 2^\omega \omega^2)$.
\end{theorem}
\begin{proof}
  (See Theorem 15, and also Theorem 5 for more context, both in \cite{Cifuentes2015}.)
  To compute the permanent from the whole tree, we have to compute the tables $Q[t]$ and then $P[t]$ for all $t \in T$, which takes the times obtained from Lemmas~\ref{lem:app:runtime-Q} and \ref{lem:app:runtime-P} for each node.

  This gives $\cO(|T| 2^\omega \omega^2)$ for the $Q$ tables. We bound $|T| = \cO(n)$, i.e. the size of the tree is linear in the number of columns (and rows) of the matrix. This is an assumption that we choose a reasonable tree decomposition: we could have a much larger tree, but at that point it would contain a lot of redundancy. Note that in every tree decomposition used in this paper, it is indeed the case that $|T| = \cO(n)$.

  For the $P$ tables, we also note the following: $\sum_{t \in T} k_t = |T| - 1$, i.e. summing the numbers of all children gives the number of non-root nodes: this is because the root is not a child. Then
  \begin{align*}
    \sum_{t \in T} \cO(k_t 2^\omega \omega^2)
    & = \cO(2^\omega \omega^2 \sum_{t \in T} k_t) \\
    & = \cO(|T| 2^\omega \omega^2) \\
    & = \cO(n 2^\omega \omega^2).
  \end{align*}
  Thus we have the bound on the time taken to compute all $Q$ and $P$ tables, and hence the permanent.
\end{proof}

\section{\label{sec:banded-matrices}Nearest-neighbour interferometers and banded matrices}

\Cref{fig:ABA-schematic} depicts the interferometer $\U$ composed of $D$ alternating layers of two-mode gates (LHS), where successive layers are shifted to allow the creation of a causal cone. We show the resulting graph $G(\U)$ in the example case of $D=2$ on the RHS.
Observe that the nearest-neighbour interaction and the alternation of different layers impose that $w_1, w_2$ can only increase by at most one for each added layer.
This is consistent with the result in~\cite{Clements2016} that depth \mbox{$D=m$} is required to generate an arbitrary unitary~$\U$.

\begin{lemma}[bandwidth]
  \label{lemma:bandwidths-of-ABA}
  The bandwidth $w$ of the unitary of a circuit composed of $D$ alternating layers, as shown in \Cref{fig:ABA-schematic}, satisfies the condition:
  \begin{equation}
    \label{eq:bandwidths-of-ABA}
    w = w_1 + w_2 + 1 \le 2D.
  \end{equation}
\end{lemma}

Note that \eqref{eq:bandwidths-of-ABA} is an inequality because some beamsplitters may act as a diagonal matrix, producing no coupling between the modes. This happens when the coupling angle $\theta = 0$ or $\pi$.

\begin{proof}[Proof of Lemma~\ref{lemma:bandwidths-of-ABA}]
  We use the help of \Cref{fig:ABA-schematic}, and we show how the widths grow as we increase the depth~$D$ by adding the appropriate beamsplitter layers. We proceed by induction on $D$. For simplicity, we assume generic coupling angles, i.e.~ignoring the cases where $\theta = 0, \frac\pi2, \pi, \dots$, because we are interested in the maximum possible connectivity between input and output modes. Likewise, we ignore the edge cases of $1 \times 1$ identity blocks at the first or last mode. These simplifications allow us to focus on the main idea, but do not obscure the full picture.

  \paragraph{\label{par:proof-depth-base-case} Base case:} $D=1$.
  A single layer is, by construction, a matrix with $2 \times 2$ blocks on the diagonal. In \Cref{fig:ABA-schematic}, in the first (rightmost) layer of the LHS, observe that input modes are partitioned into pairs of nearest neighbours, and coupling  is only possible within these pairs. Thus the output modes of the layer are also partitioned correspondingly.
  Conclude that the maximum number of nonzero elements in a row or columns is $w = 2 = 2D$.

  \paragraph{\label{par:proof-depth-induction-hypo}Induction hypothesis:} Suppose for $D$ layers, we have bandwidth $w \le 2D$.
  We add a new layer which is of different type than the previous. It also acts on disjoint pairs of nearest-neighbor modes, but it is shifted by one mode in comparison to the previous layer. Observe in \Cref{fig:ABA-schematic} that this, in general, expands the maximum reach of any given input mode by two more output modes: one upward (lower index) and one downward (higher index),\footnote{We are ignoring the edge case when there are no more modes to expand into. Of course, the inequality on bandwidth still holds.} i.e. the new bandwidth is bounded as:
  \[
    w' \le w + 2 \le 2(D+1).
  \]

  Together, paragraphs \hyperref[par:proof-depth-base-case]{a} and \hyperref[par:proof-depth-induction-hypo]{b} prove the Lemma by induction.
\end{proof}

\section{Tree decomposition manipulations}
\label{sec:app:tree-decomp-manip}

\subsection{\label{sec:app:permutations-columns}Permutations of columns}

In \Cref{sec:permutation-columns}, we show how to implement a column permutation on the tree decomposition. We prove below that this is correct:

\begin{lemma}[column permutation]
  \label{lemma:permutation-treedec-iso}
  Let $\mathrm{T} = (T,\rho,\kappa)$ be a tree decomposition corresponding to (the graph of) a matrix $M$. Let $\alpha$ be a permutation, and let $M'$ be a matrix of the same dimension as $M$ with components $M'_{i,j} = M_{i, \alpha(j)}$. This has the same columns as $M$, but their order is given by~$\alpha$.

  We construct a tree decomposition $\mathrm{T}' = (T, \rho, \kappa')$ of $M'$ with the same tree and row function by permuting the column labels for each $t \in T$:
  \[
    \kappa'(t) \deq \big\{ \bra{\alpha^{-1}(j)} \;\big|\; \bra j \in \kappa(t) \big\}.
  \]
\end{lemma}

\begin{proof}
  We show this in terms of $G(M')$, the bipartite graph of the matrix $M'$, which comes from the view of $M'$ as a sum of outer products; see \Cref{new:fig:graphical-rep}.
  The permuted matrix is:
  \begin{align*}
    M'
    & = \sum_{i,j} M'_{i,j} \ket i \! \bra j \\
    & = \sum_{i,j} M_{i,\alpha(j)} \ket i \! \bra j
      = \sum_{i,j} M_{i,j} \ket i \! \bra{\alpha^{-1}(j)}.
  \end{align*}
  Permuting the columns of $M$ to obtain $M'$ is equivalent to renaming the vertices of $G(M)$ corresponding to columns by $\alpha^{-1}$ while keeping the connectivity and edge weights the same. This action is a graph isomorphism. Renaming those same vertices as contents of the tree decomposition $\mathrm T$ to obtain $\mathrm T'$ ensures that $\mathrm T'$ is a tree decomposition of $G(M')$.
\end{proof}

\subsection{\label{sec:appendix-restr-tree-decomp}Restriction of a tree decomposition}

As introduced in \Cref{new:sec:submatrices,fig:skipping-nodes-original,fig:skipping-nodes-restriction}, we implement the taking of submatrices as restrictions of tree decompositions, intuitively, deleting labels from within the nodes of the tree. In this section, we show the formal details:

\begin{definition}[tree restriction]
  \label{def:restriction-of-treedec}
  Let $\mathrm{T} = (T, \rho, \kappa)$ be a tree decomposition of the graph $G(\U) = (V = \Rows \sqcup \Cols, E)$, and let $\Rows' \subseteq \Rows$ and $\Cols' \subseteq \Cols$ be subsets of rows and columns, respectively.
  Define the \emph{restriction} of $\mathrm{T}$ to those subsets as $\restr{\mathrm{T}}{\Rows', \Cols'} = (T, \rho', \kappa')$. It has the same tree $T$, but the contents of its nodes are given by $\rho' : T \to 2^{\Rows'}$ and $\kappa' : T \to 2^{\Cols'}$ which are defined as $\rho'(t) = \rho(t) \cap \Rows'$ and $\kappa'(t) = \kappa(t) \cap \Cols'$ for all $t \in T$.
\end{definition}

This tree restriction is a valid representation of the equivalent restriction of the matrix, that is the submatrix with the selected rows and columns. Formally:
\begin{lemma}
  \label{lemma:restriction-submatrix}
  In the above notation, the restriction $\restr{\mathrm{T}}{\Rows', \Cols'}$ is a tree decomposition of the submatrix $\restr{\U}{\Rows', \Cols'}$ of $\U$ that has rows $\Rows'$ and columns $\Cols'$.
\end{lemma}

\begin{proof}
  We have a tree decomposition $\mathrm{T} = (T, \rho, \kappa)$ of $G(\U)$, the graph of matrix $\U$. Let $\Rows' \subseteq \Rows$ be a subset of rows, and $\Cols' \subseteq \Cols$ of columns, and denote $\U' = \restr\U{\Rows', \Cols'}$. We will show that $\restr{\mathrm{T}}{\Rows', \Cols'}$ is a valid tree decomposition of $G(\U')$.

  Recall that both $\mathrm{T}$ and $\restr{\mathrm{T}}{\Rows', \Cols'}$ have the same tree $T$, but the contents of nodes in $\restr{\mathrm{T}}{\Rows', \Cols'}$ are defined as $\rho'(t) = \rho(t) \cap \Rows'$ and $\kappa'(t) = \kappa(t) \cap \Cols'$ for all $t \in T$.

  \paragraph{}
  First, check~\ref{def:treedec-axiom-vertex-cover} requiring that all vertices are included in the tree decomposition:
  \begin{subequations}
    \begin{align}
      \bigcup_{t \in T} \rho'(t)
      \label{eq:restr-rho-expand}
      & =
        \bigcup_{t \in T}
        \bigl[
        \rho(t) \cap \Rows'
        \bigr]
        =
        \left[
        \bigcup_{t \in T} \rho(t)
        \right] \cap \Rows' \\
      & \overset{\text{\ref{def:treedec-axiom-vertex-cover}}}{=}
        \Rows \cap \Rows'
        \overset{(\Rows' \subseteq \Rows)}{=}
        \Rows'.
    \end{align}
  \end{subequations}
  We omit the analogous steps to show $\bigcup_t \kappa'(t) = \Cols'$. Thus, we conclude that $\restr{\mathrm{T}}{\Rows', \Cols'}$ satisfies Axiom~\ref{def:treedec-axiom-vertex-cover} with respect to $G(\U')$.

  \paragraph{}
  Now, we show that $\restr{\mathrm{T}}{\Rows', \Cols'}$ satisfies \ref{def:treedec-axiom-edge-cover} w.r.t. $G(\U')$, requiring that each edge of the latter be contained in some tree node. Note that $G(\U')$ is an induced subgraph of $G(\U)$, meaning that for each edge $(\ket i, \bra j)$ in $G(\U)$, if its endpoints are in $G(\U')$, then the edge is in $G(\U')$ as well (with the same weight). By \ref{def:treedec-axiom-edge-cover} on $\mathrm{T}$ w.r.t. $G(\U)$, for every edge $(\ket i, \bra j)$, there is a node $t \in T$ containing it, i.e. $\ket i \in \rho(t)$ and $\bra j \in \kappa(t)$. If $\ket i \in \Rows'$ and $\bra j \in \Cols'$, then by definition of the restriction they belong to $\rho'(t)$ and $\kappa'(t)$. Hence \ref{def:treedec-axiom-edge-cover} holds for $\restr{\mathrm{T}}{\Rows', \Cols'}$ with respect to $G(\U')$.

  \paragraph{}
  We finish by showing that \ref{def:treedec-axiom-subtree} holds for $\restr{\mathrm{T}}{\Rows', \Cols'}$ w.r.t. $G(\U')$, requiring that a set of nodes containing a vertex forms a subtree. Define the preimage of $\rho$ as $\rho^{-1} : \Rows \to 2^T$ where $\rho^{-1}(\ket i) = \left\{ t \in T \,\middle|\, \ket i \in \rho(t) \right\}$ for all $\ket i \in \Rows$. Define the preimage $\kappa^{-1}$ analogously. These give us the set of nodes that contain a row or column. By Axiom~\ref{def:treedec-axiom-subtree} that holds for $\mathrm{T}$ w.r.t. $G(\U)$, each $\rho^{-1}(\bra i)$ (resp. $\kappa^{-1}(\ket j)$) forms a subtree of $T$. Under the restriction to $\Rows'$ and $\Cols'$, the preimages are clearly just restrictions (on the domain) $(\rho')^{-1} = \restr{\rho^{-1}}{\Rows'}$ and $(\kappa')^{-1} = \restr{\kappa'}{\Cols'}$. Thus we conclude that Axiom~\ref{def:treedec-axiom-subtree} is satisfied for $\restr{\mathrm{T}}{\Rows', \Cols'}$ as a tree decomposition of $G(\U')$.

  We conclude that the restriction $\restr{\mathrm{T}}{\Rows', \Cols'}$ is a valid tree decomposition, and that is it the decomposition of $G(\U')$, the graph of the submatrix $\U' = \restr\U{\Rows', \Cols'}$.
\end{proof}

An important fact, stated in eq.~(\ref{new:eq:width-of-decomposition}), is that the restriction produces a tree decomposition that is no worse than the original:
\begin{lemma}[treewidth]
  \label{lemma:treewidth-of-restriction}
  Using previous notation, the treewidth of the restriction is bounded as
  \begin{equation*}
    \tw\bigl( \restr{\mathrm{T}}{\Rows', \Cols'} \bigr)
    \le \min\bigl( \tw(\mathrm{T}), |\Rows'| + |\Cols'| - 1 \bigr).
  \end{equation*}
\end{lemma}
\begin{proof}
  A restriction may only remove vertices from the decomposition, so the treewidth can only decrease, i.e. $\tw( \restr{\mathrm{T}}{\Rows', \Cols'} ) \le \tw(\mathrm{T})$. Treewidth measures the size of the largest node, and that cannot be larger than $|\Rows'| + |\Cols'|$, i.e. the total number of vertices left. Recall from (\ref{eq:treewidth-of-treedec}) that treewidth is one less than the size of largest node. In general, we cannot know which of the choices $\tw(\mathrm{T})$ or $|\Rows'| + |\Cols'| - 1$ will be lower, hence we keep the minimum operator.
\end{proof}

Finally, we note that the tree restriction $\restr{\mathrm{T}}{\Rows', \Cols'}$ is not necessarily the best decomposition of $\restr{\U}{\Rows', \Cols'}$:  Lemma~\ref{lemma:treewidth-of-restriction} only shows the upper bound on its treewidth. However, crucially, the restriction is easy to find even when an optimal decomposition may not be, e.g. due to column permutations, or if we need a submatrix of a submatrix of~$\U$. Furthermore, as the restriction requires nothing more than taking two set intersections at each node of the tree, it is efficient to compute.

\subsection{\label{sec:redundancy-of-nodes}Redundancy of nodes}

As defined in \Cref{sec:tree-decomp-band}, each node of our tree decomposition contains a single column label. When restricting the decomposition to $\mathrm T' = (T, \rho', \kappa')$, we may end up with a node $t \in T$ without any columns, i.e. $\kappa'(t) = \varnothing$. Indeed, our algorithm will perform such restrictions extensively.
In the following Lemma~\ref{lemma:skipping-nodes}, we prove that in this case, deleting this node and connecting its neighbours in the appropriate way leads to a new tree decomposition, denoted $\mathrm T \setminus t$, that still gives the correct permanent.

\begin{definition}[removing a node]
  Let $\mathrm{T}$ be a tree decomposition with a linear tree $T$ (no branching), as is the case in our algorithm, let $t$ is an internal node (not root nor leaf), and let $p, c$ be the parent and the single child of $t$, respectively. Then $\mathrm T \setminus t$ is a decomposition without the node $t$ where $p$ is connected directly to $c$. If $t$ is the root (has no parent), then delete it and make its child $c$ the new root in $\mathrm T \setminus t$. If $t$ is a leaf, we just delete it. We present an example in \Cref{fig:skipping-nodes-redundant,fig:skipping-nodes-skipped}.
\end{definition}

Below, we show that this preserves the permanent by proving that both $\mathrm T$ and $\mathrm T \setminus t$ represent the same matrix if $\kappa'(t) = \varnothing$. Intuitively, in terms of Algorithm CP, the tables $Q[t]$ and $Q'[t | c_j]$ of node $t$ (and its child $c_j$) become trivial, and the tables $Q''[t \gets c_j]$ and $P[t]$ only act to bring data from the child of $t$ up the tree, making them accessible to the parent of $t$. Thus we may use $\mathrm{T} \setminus t$ instead of $\mathrm{T}$ in our permanent calculation.

\begin{lemma}[redundant nodes]
  \label{lemma:skipping-nodes}
  Let $\mathrm{T} = (T, \rho, \kappa)$ be a tree decomposition of the graph of a matrix $M$, where $T$ is a linear tree. Suppose there is a node $t \in T$ such that $\kappa(t) = \varnothing$ and $\rho(t) \subseteq \rho(T \setminus t)$. Then $\mathrm T \setminus t$ also represents the matrix $M$.
\end{lemma}

\begin{note}
  \label{note:rendundancy-subset}
  A subtle but important formal detail is the codition that $\rho(t) \subseteq \rho(T \setminus t)$. The decomposition $\mathrm{T} \setminus t$ is not equivalent to $\mathrm{T}$ for computing the permanent if there exists a row vertex $\ket i \in \rho(t)$ that is not contained in any other node. If such a row vertex $\ket i$ exists, and since $\kappa(t) = \varnothing$, there is no column vertex to which $\ket i$ could be connected. This corresponds to having a row of all zeros in the matrix and consequently $\per M = 0$. In such case, using $\mathrm{T} \setminus t$ we would compute the wrong (generally non-zero) permanent.

  The above problem can arise if we obtain $\mathrm{T}$ as a restriction of some larger tree decomposition, in which case we have to be careful: Suppose there is a column $\bra j$ that has a single non-zero value in row $\ket i$. A restriction that removes $\bra j$ must also remove $\ket i$.
\end{note}

\begin{proof}[Proof of Lemma~\ref{lemma:skipping-nodes}]
  We verify that $\mathrm{T} \setminus t$ satisfies the Axioms of a tree decompossition with respect to the graph $G(M)$. Recall that $G(M)$ has vertices $\Rows \sqcup \Cols$, where $\Rows$ is the set of row labels and $\Cols$ of column labels. There is an edge $(\ket i, \bra j)$ iff $M_{ij}$ is nonzero.

  \paragraph{}
  In the original tree decomposition $\mathrm{T}$, by Axiom~\ref{def:treedec-axiom-vertex-cover}, $\rho(T) = \bigcup_{q \in T}\rho(q) = \Rows$ and similarly $\kappa(T) = \Cols$. In our case, node~$t$ has no columns ($\kappa(t) = \varnothing$), so $\kappa(T \setminus t) = \Cols$ already. Removing $t$ does not violate Axiom~\ref{def:treedec-axiom-vertex-cover} on columns.

  We verify the same for rows: We have $\rho(T) = \Rows$ and $\rho(t) \subseteq \rho(T \setminus t)$. This implies that
  \[
    \Rows = \rho(T) = \rho(t) \cup \rho(T \setminus t) = \rho(T \setminus t).
  \]
  Thus $\mathrm{T} \setminus t$ satisfies Axiom~\ref{def:treedec-axiom-vertex-cover}.

  \paragraph{}
  Now, verify Axiom~\ref{def:treedec-axiom-edge-cover}, stating that every edge must be contained in a node. The decomposition $\mathrm{T}$ satisfies this. Node $t$ has $\kappa(t) = \varnothing$, so it cannot contain any edges. Therefore all edges are contained in $\mathrm{T} \setminus t$, and thus the latter satisfies Axiom~\ref{def:treedec-axiom-edge-cover}.

  \paragraph{}
  Finally, to verify Axiom~\ref{def:treedec-axiom-subtree} requiring that the set of all nodes containing a row (resp. column) forms a subtree of $T$, we note that this is already the case in $\mathrm{T}$, and that removing a node cannot invalidate this.
\end{proof}

\begin{example}[disjoint graphs]
  \label{ex:disjoint-graphs}
  As an illustrative special case, suppose that $t$ has no columns as above, and also no rows, i.e. $\rho(t) = \varnothing$. This clearly satisfies the conditions of Lemma~\ref{lemma:skipping-nodes}. Denote $T_1$ the set of ancestors of $t$ and $T_2$ the set of descendants of $t$, both subtrees of $T$. Since $t$ is completely empty, by Axiom~\ref{def:treedec-axiom-subtree}, the two trees $T_1$ and $T_2$ correspond to blocks of the matrix which share no rows nor columns. We can also view this as joining together two independent graphs (and their tree decompositions), by putting an empty node between them.
  Up to an appropriate permutation of rows and columns, corresponding to a change of basis, the matrix is block diagonal and its permanent is the product of permanents of blocks.
\end{example}

\end{document}